\documentclass[11point]{amsart}
\usepackage{amsmath, xcolor, datetime}
\usepackage{amscd}

\usepackage{amssymb,amsmath,amscd,epsfig,amsthm,enumerate,import,graphicx}
\usepackage{caption}
\usepackage{subcaption}

\usepackage{relsize}

\usepackage{hyperref}   
\hypersetup{
	colorlinks=true,  
	pdfborder={0 0 0},      
	pdfauthor={I am the Author} 
}

\usepackage{geometry}                
\newtheorem{theorem}{Theorem}[section]
\newtheorem{proposition}[theorem]{Proposition}
\newtheorem{lemma}[theorem]{Lemma}
\newtheorem{corollary}[theorem]{Corollary}
\theoremstyle{definition}
\newtheorem{definition}[theorem]{Definition}

\theoremstyle{remark}
\newtheorem{remark}[theorem]{Remark}
\newtheorem{remarks}[theorem]{Remarks}

\newcommand{\be}{\begin{equation}}
\newcommand{\ee}{\end{equation}}

\makeatletter
\newcommand{\tpitchfork}{%
  \vbox{
    \baselineskip\z@skip
    \lineskip-.52ex
    \lineskiplimit\maxdimen
    \m@th
    \ialign{##\crcr\hidewidth\smash{$-$}\hidewidth\crcr$\pitchfork$\crcr}
  }%
}
\makeatother

\newcommand{\Rightleftarrow}{\Leftrightarrow}

\newcommand{\vbar}{{\overline{v}}}

\newcommand{\SU}{\text{SU}}

\newcommand{\h}{\hbar}
\newcommand{\zbar}{\overline{z}}
\newcommand{\wbar}{\overline{w}}

\newcommand{\bbC}{{\mathbb C}}

\newcommand{\bbN}{{\mathbb N}}
\newcommand{\bbR}{{\mathbb R}}

\newcommand{\bbP}{{\mathbb P}}

\newcommand{\calD}{{\mathcal D}}
\newcommand{\calB}{{\mathcal B}}
\newcommand{\calG}{{\mathcal G}}
\newcommand{\calU}{{\mathcal U}}
\newcommand{\calQ}{{\mathcal Q}}

\newcommand{\calL}{{\mathcal L}}
\newcommand{\calM}{{\mathcal M}}

\newcommand{\calH}{{\mathcal H}}
\newcommand{\calE}{{\mathcal E}}
\newcommand{\calF}{{\mathcal F}}
\newcommand{\calC}{{\mathcal C}}
\newcommand{\calR}{{\mathcal R}}

\newcommand{\calW}{{\mathcal W}}

\newcommand{\wh}{\widehat}

\newcommand{\LCP}{\calL_{\bbC\bbP^{N-1}}}

\newcommand{\tr}{{\mbox{Tr}}}

 \newcommand{\zetabar}{\overline{\zeta}}

\newcommand{\frakm}{{\mathfrak m}}

\newcommand{\dbar}{\bar\partial}

\newcommand{\inner}[2]{\langle#1,#2\rangle}
\newcommand{\norm}[1]{\| #1\|}

\newcommand{\ket}[1]{| #1\rangle}
\newcommand{\bra}[1]{\langle #1|}

\newcommand{\delbar}{\overline{\partial}}

\renewcommand{\Box}{\square}

\newcommand{\calBk}{\calB^{(k)}}

\newcommand{\hinv}{\hbar^{-1}}

\newcommand{\PB}[2]{\{#1,#2\}}
\reversemarginpar

\newcommand{\hess}[1]{\text{Hess}(#1)}

  \DeclareSymbolFont{bbold}{U}{bbold}{m}{n}
\DeclareSymbolFontAlphabet{\mathbbold}{bbold}

 \newcommand{\C}{\bbC}

\begin{document}

\title{Reduction and coherent states}
\author{Jenia Rousseva}
\address{Mathematics Department\\
	University of Michigan\\Ann Arbor, Michigan 48109}
\email{jenia@umich.edu}
\author{Alejandro Uribe}
\address{Mathematics Department\\
University of Michigan\\Ann Arbor, Michigan 48109}
\email{uribe@umich.edu}

\maketitle
\date{}
\begin{abstract}
We apply a quantum version of dimensional reduction to Gaussian coherent states in 
Bargmann space to obtain squeezed states on complex projective spaces.
This leads to a definition of a family of squeezed spin states (definition \ref{TheDef}) with 
excellent semi-classical properties,
governed by a symbol calculus.  We prove semiclassical norm estimates and a propagation result.
\end{abstract}

\tableofcontents

\section{Introduction}

The procedure of dimensional reduction (symplectic reduction or symplectic quotient) 
is a well-established method
for studying Hamiltonian systems with symmetry.  It is also a way to construct interesting 
symplectic manifolds.  For example, the complex projective space $\bbC\bbP^{N-1}$ is
a symplectic quotient of $\bbC^N$.   The quantum version of reduction can be thought of
as a very general separation of variables.
In this paper we use quantum reduction to construct
a family of squeezed coherent states on projective spaces, in particular on $\bbC\bbP^1$,
which is to say, squeezed SU$(2)$ coherent states.  Our states have excellent semi-classical
behavior that is governed by a {\em symbol}, which we define. 

Let us begin by clarifying what we mean by quantum reduction.  We will state the definition
in the following general setting.  Let $(M,\omega)$ be a 
K\"ahler manifold, and $\calL\to M$ be a Hermitian holomorphic line bundle whose curvature
is $\omega$.  Restricting our attention, for simplicity,  to the symmetry group $S^1$, assume that
the circle acts on $\calL\to M$ preserving all structures.  We denote the corresponding momentum
map by
\[
\frakm: M\to \bbR.
\]
Let us assume that zero is a regular value of $\frakm$ and that the action of $S^1$ on
$\frakm^{-1}(0)$ is free, so that the quotient
\[
X = \frakm^{-1}(0)/S^1
\]
is a smooth manifold.  It is well known (see for instance \cite{GS}) that $X$ inherits
a K\"ahler structure and quantizing holomorphic line bundle
\[
\calL_X\to X.
\]
Moreover, $S^1$ acts by translations on the space of square-integrable holomorphic sections
(general Bargmann spaces),
\[
\calB_M= H^2(M, \calL)\cap L^2(M, \calL).
\]
The main result of \cite{GS} is that, when $M$ is compact there is a natural 
isomorphism
\begin{equation}\label{[qr]=0}
\calB_X\cong \calB_M^{S^1},
\end{equation}
where the right-hand side denotes the space of $S^1$ invariant vectors in $\calB_M$.
The isomorphism $\calB(M)^{S^1}\to\calB(X)$ is simply restriction to $\frakm^{-1}(0)$.
(We will normalize the restriction, for convenience.)  
In our main example $M=\bbC^N$ is not
compact, but (\ref{[qr]=0}) still holds.

It will be important to consider this construction for all tensor powers
\[
\calL^k\to M,\qquad k=1,\, 2,\, \ldots,
\]
where $k$ will be interpreted as $1/\h$.  Our main results are asymptotic as $k\to\infty$.

Let us introduce the notation

\begin{equation}\label{}
\calBk_M = H^0(M, \calL^k)\cap L^2(M, \calL^k)
\end{equation}
and similarly for $\calBk_X$.
Taking the $k$ tensor power of $\calL$ is equivalent to replacing the symplectic form $\omega$ by $k\omega$, so the previous isomorphism
(\ref{[qr]=0}) holds for each $k$:
\begin{equation}\label{[qr]=0k}
\calBk_X\cong \left(\calBk_M\right)^{S^1}.
\end{equation}

\medspace
The notion of quantum reduction is as follows.
\begin{definition} 
The sequence of operators
	\[
\calR_k: \calB_M^{(k)} \to \calB_X^{(k)},
\]
defined as the composition
\[
\calR_k:\	\calB_M^{(k)} \xrightarrow{P_k}\left( \calB_M^{(k)} \right)^{S^1}\cong \calB_X^{(k)},
\]
where $P_k$ is orthogonal projection (averaging), will be called the {\em quantum reduction}
operator.  
\end{definition}

One has the following general theorem (see the appendix in \cite{GS} ), which, although we will not use explicitly, explains why our
approach yields good semi-classical estimates:
\begin{theorem}
The quantum reduction operator is a Fourier integral operator 
quantizing the canonical relation
\[
\left\{(x,m)\in X\times M\;;\; m\in\frakm^{-1}(0)\ \text{and}\ \pi(m)=x\right\}\subset X\times M.
\]
\end{theorem}

\subsection{Reduction of Gaussian coherent states and first results}
In this paper, we focus on the case $M=\bbC^N$ with the symplectic form
\[
\omega =i\, dz\wedge d\,\zbar := i\sum_{j=1}^N dz_i\wedge d\zbar_i.
\]
If we write the real and imaginary parts of $z$ so that
\[
z_j = \frac{1}{\sqrt{2}}\left(q_j-ip_j\right),
\] 
then $\omega = dp\wedge dq$.
The line bundle $\calL$ is trivial but its connection is not, it is given by
\begin{equation}\label{connection}
\nabla = d + \frac 12 \left(zd\zbar -\zbar dz\right).
\end{equation}
Then, the space of square integrable functions satisfying $\nabla_{\dbar}\psi = 0$ 
is the familiar Bargmann space
\[
\calBk(\bbC^N) =  \left\{\psi = f(z)e^{-k|z|^2/2}\;;\; \dbar f = 0\right\}
\cap L^2(\bbC^N).
\]

We take $\frakm(z) = |z|^2 -1$, and the $S^1$ action on sections is
\begin{equation}\label{}
(e^{i\theta}\cdot \psi)(z) = e^{-ik\theta}\psi\left(e^{i\theta }z\right).
\end{equation}
The energy level $\frakm^{-1}(0)$ is the unit sphere $S^{2N-1} = \frakm^{-1}(0) = \left\{  z\;;\; |z|=1 \right\}$,
and the reduced space is $X=\bbC\bbP^{N-1}$.  The projection
\begin{equation}\label{hopfFib}
\pi:  S^{2N-1} \to \bbC\bbP^{N-1}
\end{equation}
is the (general) Hopf fibration.  

The Bargmann space of the quotient is customarily taken to be
\begin{equation}\label{}
\calBk_{\bbC\bbP^{N-1}} = \left\{ \text{restrictions to } S^{2N-1} \text{ of homogeneous polynomials of degree } k  \right\},
\end{equation}
with the Hilbert space structure of $L^2(S^{2N-1} )$.  Note that the natural action of $\SU(N)$ on the sphere
induces a representation on $\calBk_{\bbC\bbP^{N-1}}$, by $S_k(g)(\psi)(z) = \psi(zg)$. All these representations are irreducible. 
For $N=2$ this is the unique irreducible representation of $\SU(2)$ of dimension $k+1$.

Elements in $\calBk_{\bbC\bbP^{N-1}}$ can be thought of as sections of the $k$-th tensor power of
the hyperplane bundle over $\bbC\bbP^{N-1}$, $\calL^k\to\bbC\bbP^{N-1}$.  This point of view is
the one taken by geometric quantization, or the orbit method in representation theory.

\medskip
In the case $M=\bbC^N$ with the circle action described above, the (normalized)
reduction operator is
\begin{equation}\label{defReductOp}
\forall\psi\in\calBk_{\bbC^N},\ \forall z\in S^{2N-1}\qquad \calR_k(\psi)(z) = \frac{1}{2\pi}\int_0^{2\pi}e^{-ikt}\,\psi(e^{it}z)\, dt.
\end{equation}
It is easy to check that $\calR_k(\psi)$ is a homogeneous polynomial of degree $k$ in $z$, and
therefore it is an element in $\calBk_{\bbC\bbP^{N-1}}$.  (See Proposition \ref{prop:exactReducedState} for
an exact expression.)

\medskip
We will apply the operator  $\calR_k$ to the Gaussian coherent states in the Bargmann space
of $\bbC^N$. To recall their definition, let us introduce the generalized unit disk
\[
\calD_N = \left\{ N\times N\text{ complex symmetric matrices } A \text{ such that }  A^*A<I\right\}.
\]
We will also need a coordinate-free version of this space.  Let $\calH$ be a complex vector space
with a Hermitian inner product, and let $\calG:\calH\to\bbR$ be the standard Gaussian,
$\calG(v) = e^{-\norm{v}^2/2}$.
Let us then define
\[
\calD(\calH) = \left\{ \text{quadratic forms } Q:\calH\to\bbC \text{ such that } e^{Q/2}\,\calG\in L^2(\calH)\right\}.
\]
One can then show that $Q\in\calD(\bbC^N)$ iff the symmetric matrix $A$ associated with $Q$ in the usual sense is in $\calD_N$.

The Gaussian coherent states in Bargmann space are of the following form:
$\forall A\in\calD_N$ and $w\in\bbC^N$, let
$Q_A(z) = zA z^T$ (where $z\in\bbC^N$ is considered a row vector).  The associated state is
\begin{equation}\label{}
\psi_{A,w}(z) := e^{kQ_A(z-w)/2}\,e^{kz\wbar^T}\,e^{-k|w|^2/2}\, e^{-k|z|^2/2}.
\end{equation}
$\psi_{A,w}$ is the quantum translation of $\psi_{A,0}$ by $w$, which is called the center of $\psi_{A,w}$.
(For further discussion of quantum translations in Bargmann space see the Appendix.)

Our main objects of study in this paper are the reduced states 
\begin{equation}\label{redState}
\Psi_{A,w}:=\calR_k \left(\psi_{A, w}\right),\ w\in S^{2N-1}
\end{equation}
or, more precisely, their asymptotic properties as $k\to\infty$.
The formal definition above implies that these are functions on $S^{2N-1}$ that are restrictions
of certain polynomials of degree $k$.  As already mentioned, they can also be thought of as
sections of $\calL^k\to \bbC\bbP^{N-1}$.  

We mention right away the important example when $A=0$ (the standard or ``non-squeezed" coherent states).  One can then readily 
compute
\begin{equation}\label{}
\forall z\in S^{2N-1}\quad \Psi_{0,w}(z) = \frac{e^{-k}}{2\pi}\int_0^{2\pi} e^{-ikt}\, e^{ke^{it}z\wbar^T}\, dt =
\frac{e^{-k}\, k^k}{k!} (z\wbar^T)^k \sim \frac{1}{\sqrt{2\pi k}}(z\wbar^T)^k.
\end{equation}
When $N=2$, this is (up to a multiplicative constant) a standard spin coherent state,
see for example Chapter 7 in \cite{CR}.

\medskip
We now summarize some of our results.
\begin{theorem}\label{Main}
Let $A\in \calD_N$ and $w\in\bbC^N$ be such that $|w|=1$. 
Then $\Psi_{A,w}= \calR_k \left(\psi_{A, w}\right)$ has the following properties:
\begin{enumerate}
	\item Its micro-support (or semi-classical wave-front set) as $k\to\infty$ consists of the $S^1$ orbit
	of $w$, that is, $\{e^{it}w\;;\; t\in [0,2\pi]\}$.  Alternatively, as a section of 
	$\calL^k\to\bbC\bbP^{N-1}$ it consists of the single point
	\[
	\varpi:= \pi(w)\in\bbC\bbP^{N-1},
	\]
	where $\pi$ is the Hopf fibration (\ref{hopfFib}).

	\item  If $\eta\in \calH_w:= (\bbC w)^\bot$ (the Hermitian orthogonal space to the complex line spanned by
	$w$), one has
%
	\begin{equation}\label{whatTheSymbolIs}
	 \sigma_A(\eta):=	\lim_{k\to\infty} \sqrt{k}\,\Psi_{A,w}\left(w+ \eta/\sqrt{k} \right) = 
	\frac{1}{2\pi}\,e^{-|\eta|^2/2}  \int_{-\infty}^{\infty} e^{Q_A(isw + \eta)/2} \, e^{-s^2/2}\,ds,
		\end{equation}
		and, moreover, 
		\begin{equation}\label{symbDownstairs}
\sigma_A(\eta)		=
		\frac{1}{\sqrt{2\pi}}\frac{1}{\sqrt{Q_A(w)+1}}	e^{Q_{\rho_w(A)}(\eta )/2}\, e^{-|\eta|^2/2}
		\end{equation}
		for some $Q_{\rho_w(A)}\in\calD(\calH_w)$.  
		
			\item For all $A,\, B\in\calD_N$ one has
		\begin{equation}\label{innerEstimate}
		\inner{\Psi_{A, w}}{\Psi_{B, w}}_{\calBk_{\bbC\bbP^{N-1}}}
		 = \frac{2\pi}{k^{N}}\int_{\calH_w}\sigma_A(\eta)\,\overline{\sigma_B}(\eta)\, dL(\eta) + O(k^{-N-1})
		\end{equation}
		where $dL$ stands for Lebesgue measure.
\end{enumerate}
\end{theorem}

\begin{remark}\label{Remars1}   Some comments on the previous statements:
\begin{enumerate}
	\item Since $A\in\calD_N$ and $|w|=1$, $\Re(Q_A(w)+1) >0$.  The branch of the square root 
	in (\ref{symbDownstairs}) is the natural analytic extension to the right half of the complex plane.
	
	\item The space $\calH_w$ is in fact a subspace of $T_wS^{2N-1}$; it is the horizontal subspace
	at $w$ of the natural connection on the Hopf fibration $\pi: S^{2N-1}\to \bbC\bbP^{N-1}$.  
	The differential $d\pi_w$ induces an isometry $\calH_w\cong T_\varpi \bbC\bbP^{N-1}$, where the latter
	space is given the Fubini-Study metric.  We will tacitly use this identification in what follows.

	\item The rescaling in the argument of $\Psi_{A,w}$ in the left-hand side of (\ref{whatTheSymbolIs}) 
	is an example of the {\em local scaling asymptotics} in semi-classical
	analysis of quantized K\"ahler manifolds, first introduced by Bleher, Shiffman and Zelditch in \cite{BSZ}.  
	We will expand on the meaning of this in \S \ref{subsec:GeomSymb}.
	
	\item In case $w=(1,\vec{0})$ (the general case can be reduced to this by the action of a unitary matrix),
	one has that $\rho_w(A)$ is the lower $(N-1) \times (N-1)$ principal minor of 
\begin{equation}\label{}
	A- \frac{Aw^TwA}{wAw^T+1}.
\end{equation}

	\item Theorem \ref{thm:stationary_proj} gives the asymptotic behavior of (\ref{redState}) {\em at} $\varpi$.

\end{enumerate}
\end{remark}

\medskip
The function $\sigma_A$ is the main invariant associated with $\Psi_{A,w}$, so we give it a name.

\begin{definition}
	The function $\sigma_A: \calH_w\to \bbC$ given by the expressions (\ref{whatTheSymbolIs}) and (\ref{symbDownstairs})  
will be considered as a function of the Bargmann space of the tangent space $T_\varpi\bbC\bbP^{N-1}$ (with $\h=1$), and
	will be called the {\em symbol} of (\ref{redState}).
\end{definition}

\begin{remark}\label{Remark2}
	It is very convenient to extend by linearity the definition of symbols of 
	reduced states at the same center $w\in S^{2N-1}$.  We will also agree that multiplying 
	$\Psi_{A,w}$ by a power of $k$ results in a function having the same symbol as $\Psi_{A,w}$.
	
	Note that the symbol of a standard spin coherent state is simply the Gaussian \\
	$\sigma_{A=0}(\eta) = \frac{1}{\sqrt{2\pi}}e^{-|\eta|^2/2}$.
\end{remark}

\medskip
We will also prove a propagation theorem for reductions of
Gaussian coherent states under suitable quantum Hamiltonians.  
The symbols of the propagated states are computed 
in an entirely analogous way as in the Euclidean case, that is, using the metaplectic representation.  
These results are presented in \S 5, see 
Theorems \ref{thm:FirstPropaThm} and \ref{thm:SecondPropaThm}.

\subsection{An explicit formula}

One can compute an exact algebraic expression for the reduced states (\ref{redState}), which may be useful
for numerical calculations.

\begin{proposition}\label{prop:exactReducedState}
	 For all $z, w\in S^{2N-1}$ one has:
\begin{equation}\label{exactRedState} 
\Psi_{A,w}(z) = e^{-k}\, e^{kQ_A(w)/2}  \sum_{\ell\geq k/2}^{k} \frac{k^\ell}{(k-\ell)!(2\ell-k)!} 
\left(\frac{1}{2}Q_A(z)\right)^{k-\ell}\,
\left( z(\wbar^T-Aw^T)\right)^{2\ell-k}.
\end{equation}
\end{proposition}
\begin{proof}
	Since
	\[
	Q_A(z-w)= Q_A(z)-2zAw^T + Q_A(w),
	\]
	we can re-write
	\begin{equation}\label{rewrite}
	\psi_{A,w}(z) = e^{-k}\, e^{kQ_A(w)/2}\,
	e^{k Q_A(z)/2} \, e^{kz(\wbar^T-Aw^T)}
	\end{equation}
	Therefore 
	\[
	\psi_{A,w}(e^{it}z) = e^{-k}\,e^{kQ_A(w)/2}\,
	\sum_{\ell =0}^\infty \frac{k^\ell}{\ell !}
	\left(e^{2it}Q_A(z)/2 + e^{it}z(\wbar^T-Aw^T) \right)^\ell .
	\]
	Now apply the binomial theorem to the $\ell$-th term of the series:
	\[
	\left(e^{2it}Q_A(z)/2 + e^{it}z(\wbar^T-2Aw^T) \right)^\ell = \sum_{j=0}^\ell
	{\ell\choose j}\, e^{it(j+\ell)}\,\left(Q_A(z)/2\right)^j\, 
	\left( z(\wbar^T-Aw^T)\right)^{\ell-j}.
	\]
	When we multiply by $e^{-ikt}$ and integrate over $t\in[0,2\pi]$ 
	only the terms with $j+\ell = k$ survive.  For each $\ell$
	there exists exactly one such term precisely when $0\leq k-\ell\leq\ell$.
	This gives the range $k/2\leq\ell\leq k$, and the expression (\ref{exactRedState}) follows.
\end{proof}

\begin{remark}
 The previous expression is exact but is ``redundant to leading order" because 
the mapping 
\begin{equation}\label{redOfSymbols}
\calD_N\ni A\mapsto \rho_w(A)\in \calD_{N-1} 
\end{equation}
is not injective, and the symbol controls the reduced state to leading order.

 Note that the case $A=0$ (standard coherent
states in the Bargmann space of $\bbC^N$), up to a multiplicative constant, the reduced state is indeed just the standard SU$(N)$ state $(z\wbar)^k$.
\end{remark}

\subsection{Squeezed spin coherent states}
The case $N=2$ is of particular interest because it corresponds to SU$(2)$, or spin-squeezed coherent states.
We next present 
an expression that agrees asymptotically with (\ref{exactRedState}) and that involves a single parameter,
$\mu\in D_1$.
We will write this approximation in a standard orthonormal basis of $\calBk_{\bbC\bbP^{1}}$,
\begin{equation}\label{ketn}
\ket{n} =  \frac{1}{\pi}\, \sqrt{\frac{k+1}{2}}\, \sqrt{\binom{k}{n}} \: z_1^n\,z_2^{k-n},
\quad 0\leq n\leq k.
\end{equation}
This is a basis of eigenvectors of the operator corresponding to $\sigma_3 =\frac 12 \begin{pmatrix}
1 & 0 \\ 0 & -1
\end{pmatrix}$, the eigenvalue associated with $\ket{n}$ being $n-\frac{k}{2}$.
By equivariance of the construction under the action of SU$(2)$, it suffices to 
write the approximation in the case $w=(1,0)$.

\begin{proposition}\label{n2-case}
Let  $\mu\in\bbC$, $|\mu|<1$ and
$
[\mu] := \begin{pmatrix}
0 & 0 \\
0 & \mu
\end{pmatrix}.
$
Then
\begin{equation}\label{ketMu}
\Psi_{[\mu], (1,0)} = 
\pi k^k e^{-k} \,  \sqrt{\frac{2}{(k+1)!}}
\sum_{0\leq\ell\leq k/2}\left(\frac{1}{2k}\right)^{\ell}\frac{1}{\sqrt{(k-2\ell)!}}
\sqrt{2\ell\choose \ell}\,\mu^{\ell}\,\ket{k-2\ell}.
\end{equation}
Furthermore, for any $A= \begin{pmatrix}
a & c \\ c & b
\end{pmatrix} \in\calD_2$, if we let
\begin{equation}\label{}
\mu = \rho_{(1,0)}(A)= b-\frac{c^2}{1+a},
\end{equation}
then $|\mu|<1$ and one has
\begin{equation}\label{}
\Psi_{A,w}(z)=\Psi_{[\mu] ,(1,0)}(1 + O(1/\sqrt{k})),
\end{equation}
where the error estimate is in norm.
\end{proposition}

Identifying $\calH_{(1,0)}$ with 
the $z_2$ complex plane one finds that
\begin{equation}\label{symbolKetMu}
\sigma_{[\mu]}(z_2) = \frac{1}{\sqrt{2\pi}} e^{\mu z_2^2/2}\, e^{-|z_2|^2/2},
\end{equation}
and therefore, by (\ref{innerEstimate}), after some calculations we obtain
\begin{equation}\label{afterSome}
\norm{\Psi_{[\mu],(1,0)}} \sim \frac{\sqrt{\pi}}{ k } (1-|\mu|^2)^{-1/4}.
\end{equation}
We now proceed to normalize (\ref{ketMu}):
\begin{lemma}
The wavefunction
\begin{equation}\label{ketMuNormalized}
\ket{o,\mu} := 
\sum_{0\leq\ell\leq k/2}\left(\frac{1}{2k}\right)^{\ell}\frac{(2\ell)!}{\ell!}
\sqrt{k\choose 2\ell}\,\mu^{\ell}\,\ket{k-2\ell}
\end{equation}
agrees to leading order with $\frac{k}{\sqrt{\pi}}\Psi_{[\mu], (1,0)}$, and its norm satisfies
\begin{equation}\label{normKet}
\bra{o,\mu}\ket{o,\mu} = (1-|\mu|^2)^{-1/2} + O(1/k).
\end{equation}
(Here $o = \pi(1,0)\in\bbC\bbP^1$.)
\end{lemma}
\begin{proof}
Begin by multiplying $\Psi_{[\mu], (1,0)}$ by $k/\sqrt{\pi}$.  By (\ref{afterSome}) the result has a norm squared
that asymptotically is given by (\ref{normKet}).   Then apply Stirling's formula and simplify. 
\end{proof}

We have plotted the magnitudes of the components of the $\ell^2-$normalized $\ket{o,\mu}$ for $\mu=3/4$ and $k=30$ in Figure \ref{fig:ketmucoeffs}. 

\begin{figure}[h!]
	\centering
	\includegraphics[width=0.6\linewidth]{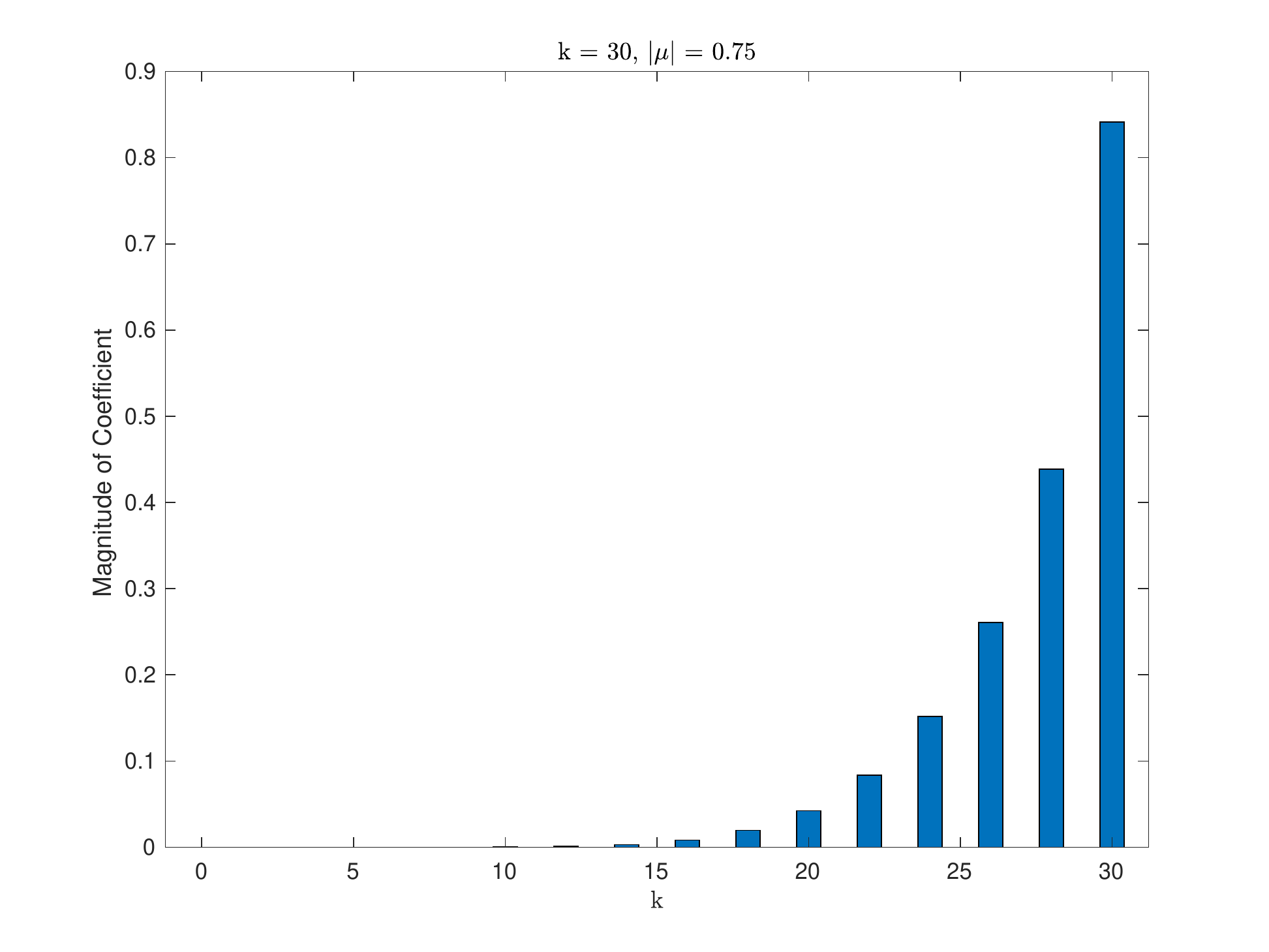}
	\caption{Plot of the components of the $\ell^2-$normalized kets \eqref{ketMuNormalized} for $k=30$ and $\mu =3/4$. Observe that the magnitudes are decaying as $k$ decreases. The closer $|\mu|$ is to zero, the more rapidly the decay of the components.}
	\label{fig:ketmucoeffs}
\end{figure}

\medskip
For future reference, note that then the symbol of $\ket{o,\mu}$ (see Remark \ref{Remark2}) is
\begin{equation}\label{ketSymbol}
\frac{1}{\sqrt{2}\, \pi}\, e^{\mu z_2^2/2}\, e^{-|z_2|^2/2}.
\end{equation}

We now let $\SU(2)$ act on the previous states:
\begin{definition}\label{TheDef}
Let $S_k:\text{SU}(2)\to \calU\left(\calBk_{\bbC\bbP^{1}}\right)$ be the natural representation
of SU$(2)$ in $\calBk_{\bbC\bbP^{1}}$.  
For any $p\in\bbC\bbP^1$, let $g\in \text{SU}(2)$ be 
such that $p =g\cdot o$.  If $\mu\in\calD_1$, let
\begin{equation}\label{generalMuState}
\ket{p,\mu} = S_k(g)(\ket{o,\mu}).
\end{equation}
We call any such state {\em a squeezed SU$(2)$ Gaussian state with center $p$ and parameter $\mu$}.
\end{definition}
We note that the notation (\ref{generalMuState}) is ambiguous, since $g$ is not unique for
a given $p$, but the ambiguity is a unitary factor (the squeezed coherent states are properly labeled by
points on $S^{3}$).

It is worthwhile to give a different description of the $\ket{o,\mu}$.
Using a trivialization of the Hopf fibration $S^3\to\bbC\bbP^1$, 
one can identify the Bargmann space of $\bbC\bbP^1$ with the space 
\begin{equation}\label{}
\calB_{\bbC\bbP^1}^{(k)} \cong \left\{ \Psi(\zeta)  = \frac{f(\zeta)}{(1+|\zeta|^2)^{k/2}} \; |\; \delbar f = 0 \right\} \cap L^2(\bbC, dm)
\quad \text{where}\quad dm = \frac{2\pi}{i}\frac{d\zeta\dbar{\zeta}}{(1+|\zeta|^2)^2}.
\end{equation} 
One can check that, in the above, $f$ must be a polynomial of degree at most $k$ in the complex variable $\zeta$.
The identification is simply by pulling back elements in $\calB^{(k)}_{\bbC\bbP{^1}}$ by the
section $S_\varpi: \bbC\to S^3$ given by
\[
S_\varpi(\zeta) = \frac{1}{\sqrt{1+|\zeta|^2}}(1,\zeta).
\]
It is not hard to compute that 
\begin{equation}\label{}
S_\varpi^*\ket{n} = \frac{1}{\pi} \,\frac{1}{(1+|\zeta|^2)^{k/2}}  \, \sqrt{\frac{k+1}{2}} \, \sqrt{\binom{k}{n}} \, \, \zeta^{k-n}
\end{equation}
and 
\begin{equation}\label{}
S_\varpi^*\ket{o,\mu} =  \frac{1}{\pi} \, \frac{k!}{(1+|\zeta|^2)^{k/2}} \sqrt{\frac{k+1}{2}}  \sum_{0\leq\ell \leq k/2} \left(\frac{1}{2k}\right)^\ell  \frac{1}{\ell ! (k-2\ell)!} \, \mu^\ell \, \zeta^{2\ell}.
\end{equation}
Figure \ref{fig:muKetsExample} shows the Husimi function $\left|S_\varpi^*\ket{o,\mu}  \right|^2$ of the ket $\ket{o,\mu}$ and its level sets
as a function of $\zeta$, for a choice of $\mu$ and $k$.

\begin{figure}[h!]
	\begin{subfigure}[b]{0.4\textwidth}
		\includegraphics[width=\textwidth]{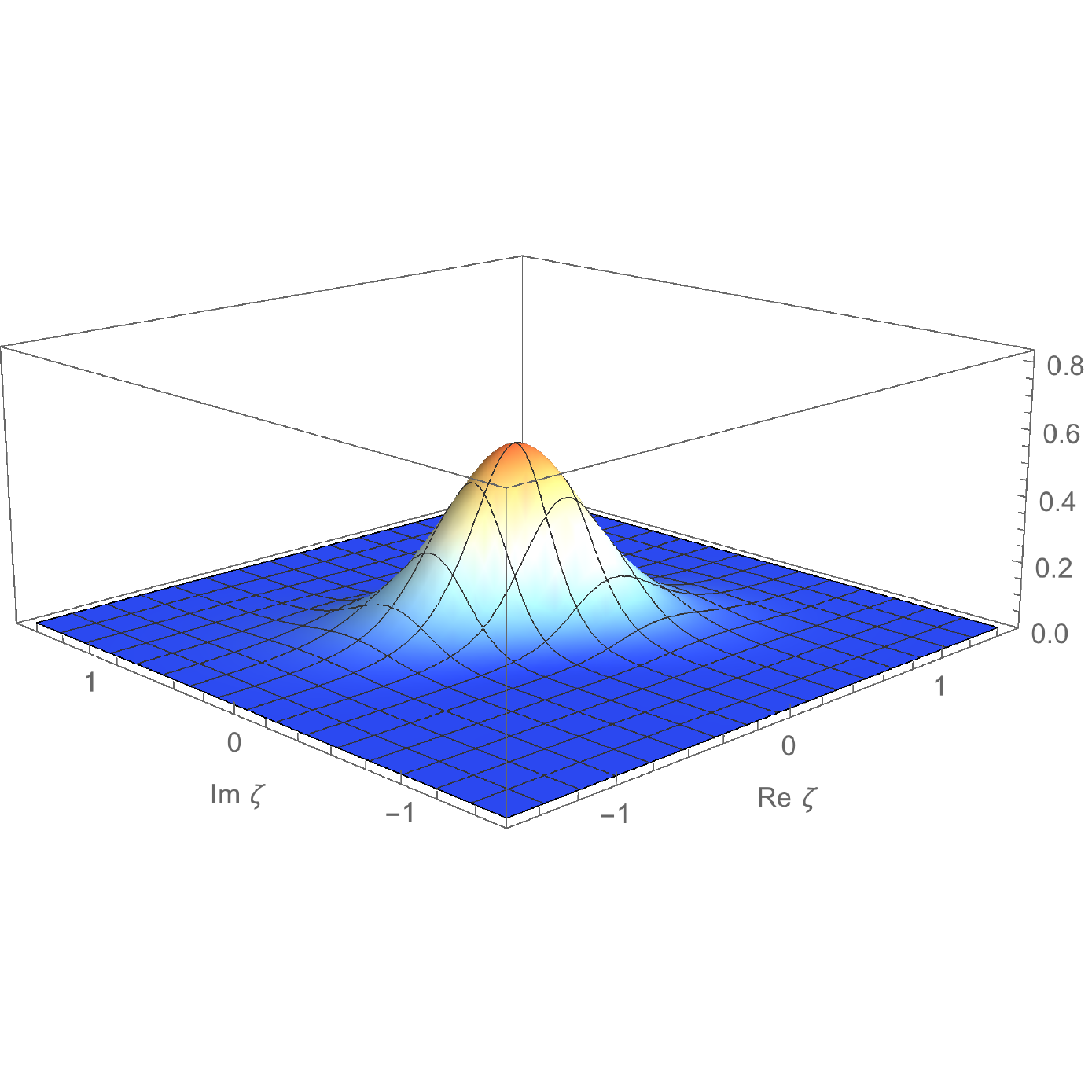}
	\end{subfigure}
	\qquad
	\begin{subfigure}[b]{0.4\textwidth}
		\includegraphics[width=0.85\textwidth]{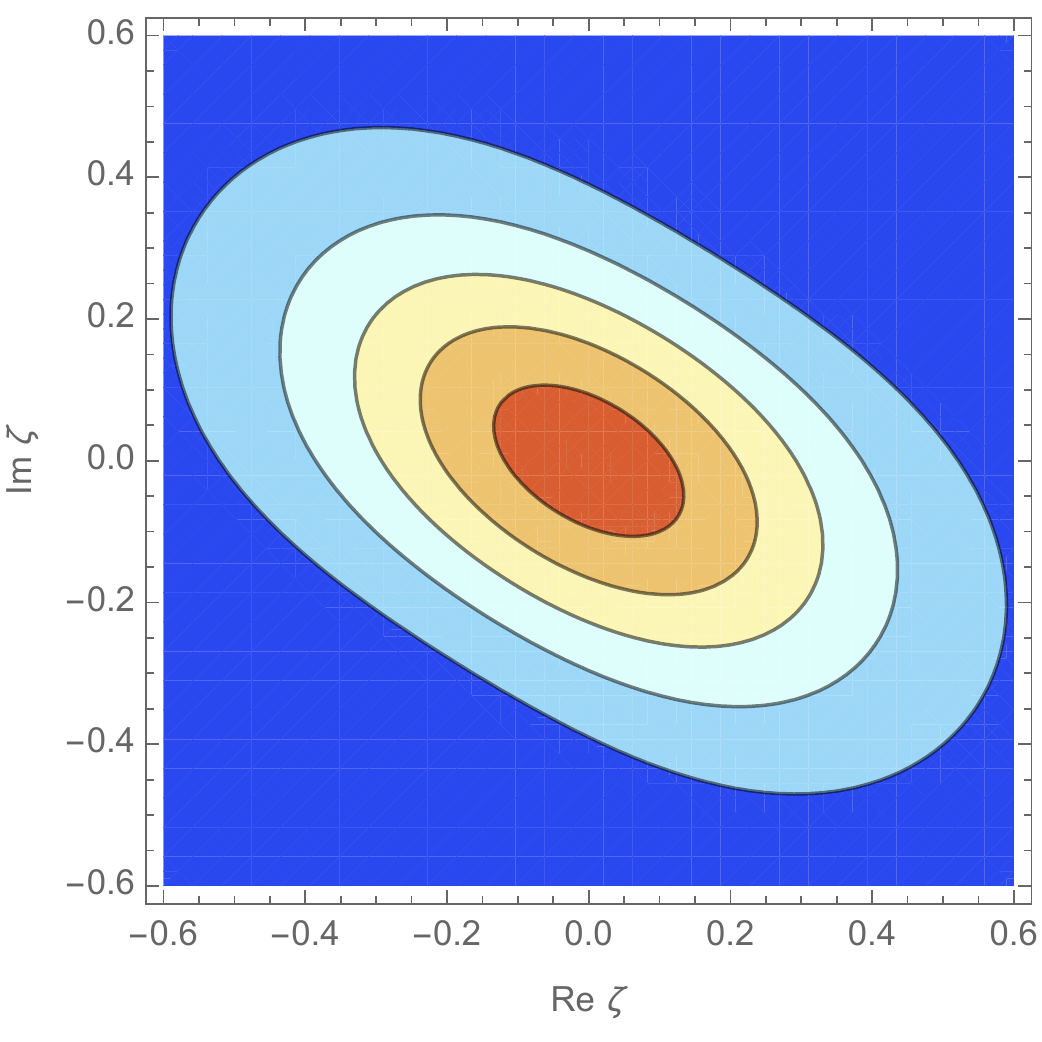}
	\end{subfigure}
	\caption{Plot of $\left|S_\varpi^*\ket{o,\mu}  \right|^2$  and its levels sets for $k=10$ and $\mu = 1/4 +i/2$.}
	\label{fig:muKetsExample}
\end{figure}

\bigskip
The paper is organized as follows.  In \S 2 we prove parts (1) and (2) of Theorem \ref{Main}, and 
in \S 3 we prove part (3) of Theorem \ref{Main}.  In \S 4 we discuss reduction, with $\h=1$, of states in the 
{\em linear} setting, which allows us to say that ``the symbol of the reduction is
the reduction of the symbol". We also recall how the metaplectic representation is constructed in Bargmann spaces,
following a paper by I. Daubechies.  This is used in \S 5 is where we prove our propagation results.

\medskip\noindent
{\em Acknowledgments:}  We wish to thank Eva Maria Graefe for calling our attention to the
problem of systematically constructing squeezed SU$(2)$ coherent states, and to her and Robert Littlejohn for useful discussions during an IMA workshop in
the summer of 2018.

\section{First estimates}

\subsection{Remarks on Gaussian states}

\subsubsection{Estimates}  We begin by establishing some fundamental estimates on Gaussian states.
\begin{lemma}
Let $A\in \calD_N$.  Then 
\begin{equation}\label{canConclude}
\exists \kappa\in [0,1)\ 
\forall z\in\bbC^N\qquad |Q_A(z)| \leq \kappa |z|^2.
\end{equation}
\end{lemma}
\begin{proof}
	Let $A\in \calD_N$.  By Takagi's factorization, there exists a unitary matrix $U$ and a diagonal matrix
	$D$ such that $A = UDU^T$, and $D$ is diagonal with entries $\kappa_j(A)\geq 0$, $j=1,\dots,N$, the 
	square roots of the eigenvalues of $A^*A$.  Let $z\in\bbC^N$ and $\gamma = zU$.  Then
	\[
|Q_A(z)| = |Q_D(\gamma)| = \left|\sum_{j=1}^N \gamma_j^2\kappa_j\right|\leq \kappa |\gamma|^2 = \kappa |z|^2,
	\]
	 where $\kappa = \max_j \kappa_j$. The assumption that $A\in\calD_N$ is equivalent to $\kappa <1$.
\end{proof}

 In particular $\Re \left(Q_A(z) \right)\leq \kappa |z|^2$ .  On the other hand, 
\begin{equation}\label{alternativePhase}
\psi_{A, w}(z) = e^{kQ_A(z-w)/2}\,e^{-k|z-w|^2/2}\,e^{ik\Im(z\wbar^T)},
\end{equation}
where $\omega$ is the symplectic form
\[
\omega(z, w) = \Im(z \wbar^T).
\]
Therefore, the Husimi function of $\psi_{A,w}$ is equal to
\begin{equation}\label{husimi}
|\psi_{A,w}|^2(z) = e^{k[\Re( Q_A(z-w)) - |z-w|^2] } \leq e^{-k[(1-\kappa) |z-w|^2] }.
\end{equation}
Since $\kappa <1$
the phase in (\ref{husimi}) is non-positive and is zero precisely at $z=w$.  
Away from $w$ the Husimi function is exponentially decreasing.
From this it follows that the semi-classical microsupport of $\psi_{A,w}$ is $\{w\}$.

The proof of the previous lemma can easily be modified to show the
equivalence of the two definitions of $\calD_N$ and $\calD(\bbC^N)$.

As another observation, we note:
\begin{lemma}\label{Overlap}
Given $A,\,B\in\calD_N$ and $v,w\in\bbC^N$, then
\[
v\not=w\quad\Rightarrow\quad \inner{\psi_{A, w}}{\psi_{B,v}} = O(k^{-\infty}).
\]
\end{lemma}
\begin{proof}
Let us write
$
	\inner{\psi_{A, w}}{\psi_{B,v}} = \int_{\bbC^N} e^{\varphi(z,\zbar)}\, dL(z)
$
	where
\[
	\varphi = Q_A(z-w)/2 + \overline{Q_B(z-v)}/2 + z\wbar^T + \zbar v^T - |v|^2/2 - |w|^2/2 -|z|^2.
	\]
Let us look for critical points of the phase.  Note that
\begin{align}
	\frac{\partial\varphi}{\partial z} &= (z-w)A + \wbar -\zbar,\quad\text{and}\label{partialz}\\
	\overline{\frac{\partial\varphi}{\partial \zbar}} &=(z-v){B} + \vbar -\zbar\label{partialzbar}.
\end{align}
{\em Claim:}  If $A\in\calD_N$, the mapping $\bbC^N\ni z\mapsto zA- \zbar\in\bbC^N$ is bijective.

{\em Proof of the claim.}  Since the map is $\bbR$-linear, it is enough to prove that its kernel is zero.
Note that
\[
zA = \zbar \quad\Rightarrow\quad \zbar\overline{A} = z \quad\Rightarrow\quad z A\overline{A} = z.
\]
Since $A$ is symmetric this means that $zAA^* = z$.  Since $A\in\calD_N$, 1 is not an eigenvalue of $AA^*$, and 
therefore $z=0$.

Since (\ref{partialz}) being equal to zero is equivalent to $zA -\zbar = wA-\wbar$,  we see that $\frac{\partial\varphi}{\partial z}=0$
iff $z=w$.  Similarly, $\frac{\partial\varphi}{\partial \zbar}=0$ iff $z=v$.  So if $v\not=w$ the phase $\varphi$ does not
have any critical points.
\end{proof}

\subsubsection{Covariance}
The Gaussian states in the Bargmann space of $\bbC^N$ 
have the following useful covariance property.  The group U$(N)$ acts on $\bbC^N$ on the
right (since we are working with row vectors), which induces an action (representation) on $\calBk_{\bbC^N}$
given by
\begin{equation}\label{}
\forall g\in \text{U}(N), \quad \psi\in\calBk_{\bbC^{N}}\qquad 
(g\cdot \psi)(z) := \psi(zg).
\end{equation}
The following is straightforward, and is very useful:
\begin{lemma}
	One has
	\begin{equation}\label{covariance}
	g\cdot \psi_{A, w} = \psi_{gAg^T, wg^{-1}}.
	\end{equation}
\end{lemma}

\subsection{Pointwise estimates of the reduced states}  Let $A\in\calD_N$ and $w\in S^{2N-1}$.
We now obtain a point-wise estimate of $\Psi_{A,w}$. 

From the definition (after a short calculation),
\begin{equation}\label{projection_int}
\forall z\in S^{2N-1}\qquad\Psi_{A,w}(z) = \frac{1}{2\pi} \int_{0}^{2\pi} e^{k\varphi(z,t)}\, dt
\end{equation}
where the phase is
\begin{equation}
\varphi(z, t) := e^{it}z \wbar^T + \frac{1}{2}(e^{it}z - w)A(e^{it}z-w)^T  -i t - \frac{1}{2}(|z|^2 + |w|^2).
\label{eq:phi}
\end{equation}
\begin{lemma} The phase satisfies $\Re(\varphi) \leq 0$.  Moreover, 
	its critical points (with respect to $t$) satisfying $\Re(\varphi)=0$ are precisely the
	solutions of  $e^{it}z =w$.
\end{lemma}
\begin{proof}
	We already know from (\ref{canConclude}) that $\Re(\varphi)=0$ iff $e^{it}z=w$.
	On the other hand, the critical points of the phase are solutions of
	\begin{equation}\label{critPts}
	e^{it}z\wbar^T + (e^{it}z-w)A z^T=1.
	\end{equation}
	This is indeed satisfied if $e^{it}z = w$.
\end{proof}
As a corollary of the previous Lemma, regarded as a section of the tensor powers of the reduced (or hyperplane) line bundle
\[
\calL^k\to\bbC\bbP^{N-1},
\]
$\Psi_{A,w}$ and all its derivatives are 
rapidly decreasing away from the point $\varpi = \pi(w)$.  This is item (1) in Theorem \ref{Main}.  
To evaluate $\Psi_{A,w}$ 
asymptotically at $\varpi$, let us apply the method of stationary phase (Theorem 7.7.5 in \cite{H}) to (\ref{projection_int}). 
Thus, assume that $e^{it_0}z = w$ for some $t_0$.  The second derivative of the phase at $t=t_0$ (see the left-hand side
of (\ref{critPts})) is equal to $i(1+wAw^T)$.  This implies:
\begin{theorem}
	With the previous notation,
	\begin{equation*}
	\Psi_{A,w}(e^{-it_0}w)
	= \frac{1}{\sqrt{2 \pi k}}
	\frac{e^{-ikt_0}}{\sqrt{w A w^T +1}} +O(k^{-3/2})\quad \text{ as } k \to \infty.
	\end{equation*}
	\label{thm:stationary_proj}
\end{theorem}

\section{Symbols}

In this section we prove the remainder of Theorem \ref{Main}.  Then, we will
place the definition of the symbol of the reduced states in a geometric context.

\subsection{Proof of part (2) of Theorem \ref{Main}}

We will now prove (\ref{whatTheSymbolIs}).  Fix $w\in S^{2N-1}$ and
$\eta\in \calH_w$, and introduce the notation
\begin{equation}\label{}
\Upsilon_A(\eta, k):= \Psi_{A,w}\left[w+\frac{\eta}{\sqrt{k}}\right].
\end{equation}
We need to show that
\begin{equation}\label{wTS}
\sqrt{k}\,\Upsilon_A(\eta, k) = 
\frac{1}{2\pi}e^{-|\eta|^2/2}  \int_{-\infty}^{\infty} e^{Q_A(isw + \eta)/2} \, e^{-s^2/2}\,ds + O(1/\sqrt{k}). 
\end{equation}

Note that
\[
\Upsilon_A(\eta, k) =  \frac{1}{2 \pi} \int_{-\pi}^{\pi} \psi_{A,w}\left(e^{it} (w+\eta/\sqrt{k}) \right) e^{-ikt} dt .
\]  
For each $k$ we split the domain of integration into three parts,
\[
\Upsilon_A(\eta,k) = \frac{1}{2\pi}\int_{-\pi}^{-a_k} + \frac{1}{2\pi}\int_{-a_k}^{a_k} + \frac{1}{2\pi}\int_{a_k}^{\pi} =: I_1+I_2+I_3,
\]
respectively, where $(a_k)$ is a sequence of positive numbers tending to zero that we will
specify later.  In particular, we will choose this sequence so that $I_1$ and $I_3$ are negligible with
respect to $I_2$.

\medskip
First let us estimate $I_3$.  Recall that $\vert \psi_{A,w}(z) \vert \leq e^{-Ck |z-w|^2}$ 
with $C = (1-\kappa)/2\in (0,1/2]$, where $\kappa$ is the largest eigenvalue of $A^*A$ (see (\ref{husimi})). 
Therefore 
\begin{align*}
\bigg\vert \psi_{A,w}\left(e^{it} (w+\eta/\sqrt{k}) \right) e^{-ikt}  \bigg\vert  
\leq e^{-Ck \vert e^{it}(w +\eta/\sqrt{k}) - w|^2} 
\leq e^{-C k|e^{it}-1|^2} 
\end{align*}
where we have used that 
$\eta \cdot \wbar =0$ and $|w|^2=1$.  Hence, 
\[
|I_3| \leq \frac{1}{2\pi} \int_{a_k}^\pi e^{-Ck|e^{it}-1|^2}  dt  \leq  C_1 \max_{t \in [a_k, \pi]} e^{-Ck|e^{it}-1|^2} 
= C_1 \: e^{-Ck|e^{ia_k}-1|^2}.
\]
Since $|e^{it} -1|^2 = t^2 + t^4 R(t)$ for some function $R(t)$
bounded in a neighborhood of zero, we conclude
\begin{equation}\label{I1andI2}
|I_3| \leq C_1\, e^{-Cka_k^2(1+a_k^2R(a_k))}
\end{equation}
and similarly for $I_1$.  
We now pick 
\begin{equation}\label{thea_k}
a_k = \left(\frac{\log(k)}{Ck}\right)^{1/2}
\end{equation}
with $C$ the above constant.  Therefore
\begin{equation}\label{orderI3}
I_3 = O\left( 1/k\right)\quad\text{and similarly for } I_1.
\end{equation}

We now turn to
$I_2 =  \frac{1}{2 \pi} \int_{-a_k}^{a_k} \psi_{A,w}\left(e^{it} (w+\eta/\sqrt{k}) \right) e^{-ikt} dt$.
After some algebra, one finds that this integral has the following form:
\begin{equation}\label{I2}
I_2= \frac{1}{2 \pi} \int_{-a_k}^{a_k}  e^{k\phi(t) + \sqrt{k}\psi(t) + \varrho(t)} dt
\end{equation}
where 
\begin{align*}
\phi(t) &= \frac{1}{2}Q_A(w)(1-2e^{it}+e^{2it}) + e^{it}-it-1 \\
\psi(t) &= (e^{2it} - e^{it})\eta A w^T \\
\varrho(t) &= \frac{1}{2} (e^{2it}Q_A(\eta)-|\eta|^2).
\end{align*}
The only critical point of the phase is at $t=0$. After a Taylor expansion at zero, one obtains:
\begin{lemma}
	The integral $I_2$ is of the form
	\begin{equation}\label{I2bis}
	I_2 = \frac{e^{[Q_A(\eta)-|\eta|^2]/2}}{2\pi}\int_{-a_k}^{a_k}
	e^{f_k(t) + g_k(t)}\, dt
	\end{equation}
	with
	\begin{equation}\label{}
	f_k(t) = -k(Q_A(w)+1)t^2/2 + it \sqrt{k}\, \eta\, A w^T  
	\end{equation}
	and
	\begin{equation}\label{}
	g_k(t) = it^3 k G(t) + t^2\sqrt{k} H(t) + itF(t),
	\end{equation}
	where $F,\, G,\, H$ are smooth $k$-independent
	functions (in particular bounded in a neighborhood of zero).
\end{lemma}

We now make the change of variables $t=s/\sqrt{k}$ in (\ref{I2bis}), to obtain
\begin{equation}\label{toObtain}
I_2 = \frac{e^{[Q_A(\eta)-|\eta|^2]/2}}{2\pi\sqrt{k}}\int_{-\infty}^{\infty}
e^{-(Q_A(w)+1)s^2/2 + is\, \eta\, A w^T}\,
e^{g_k(s/\sqrt{k})}\,\chi\left(\frac{ s}{\sqrt{k}\,a_k}\right)\, ds
\end{equation}
where $\chi$ is the characteristic function of $[-1,1]$.
We claim that 
\begin{equation}\label{bdedConv}
e^{g_k(s/\sqrt{k})}\,\chi\left(\frac{ s}{\sqrt{k}\,a_k}\right) \ \text{is uniformly bounded and converges to one }  
\ \forall s\in\bbR.
\end{equation}
To see this, observe first that
the support of $\chi\left(\frac{ s}{\sqrt{k}\,a_k}\right)$
is equal to the set of $s$ such that
\begin{equation}\label{range(s)}
|s|\leq C^{-1/2} \log(k)^{1/2},
\end{equation}
which inequality implies that $\frac{|s|^j}{\sqrt{k}}\leq  \frac{\log(k)^{j/2}}{\sqrt{k}} $
for $j=0,1\ldots$, since $C<1$.
Then, since
\begin{equation}\label{theg_ks}
g_k(s/\sqrt{k})  = \left[i s^3 G(s/\sqrt{k}) + s^2 H(s/\sqrt{k}) +is F(s/\sqrt{k})\right]
\frac{1}{\sqrt{k}},
\end{equation}
for all $s$ in the support of $\chi\left(\frac{ s}{\sqrt{k}\,a_k}\right)$
$g_k(s/\sqrt{k})$ is uniformly bounded by a constant times $\frac{\log(k)^{3/2}}{\sqrt{k}}$, which
tends to zero.

By (\ref{bdedConv}) and the Lebesgue dominated convergence theorem, 
$\sqrt{k}I_2$ converges to the right-hand side of (\ref{wTS}).
It remains to estimate the rate of convergence.
Let us define
$\calE(s,k) := e^{g_k(s/\sqrt{k})}-1$
for $s$ satisfying (\ref{range(s)}) and zero otherwise, so that
\[
e^{g_k(s/\sqrt{k})}\chi\left(\frac{ s}{\sqrt{k}\,a_k}\right) = 
\chi\left(\frac{ s}{\sqrt{k}\,a_k}\right) \left[1+ \calE(s,k)\right].
\]
Applying Taylor's theorem to $|\calE|^2$ near $s=0$, for each $k$, one gets
\begin{equation}\label{}
|\calE(s,k)|^2 = \frac{2s}{\sqrt{k}}\Re
\left[ g'_k(b/\sqrt{k})\left(
e^{\overline{g_k}(b/\sqrt{k})}-1
\right)
\right]\leq \frac{2|s|}{\sqrt{k}}\,
\left| g'_k(b/\sqrt{k})\left(
e^{\overline{g_k}(b/\sqrt{k})}-1
\right)
\right|
\end{equation} 
for $|s|<C^{-1}\log(k)^{1/2}$ and
where $b = b(s)$ is between zero and $s$, and therefore $|b(s)|\leq C^{-1}\log(k)^{1/2}$.
From this and (\ref{theg_ks}) it follows that 
\[
\left| g'_k(b/\sqrt{k})\right|\leq\frac{C_1}{\sqrt{k}}\quad\text{and}\quad 
\left|e^{\overline{g_k}(b/\sqrt{k})}-1 \right|\leq C_2
\]
for some constants $C_j>0$, for each $s$ satisfying (\ref{range(s)}).  Therefore $\exists C_3>0$ such that
\begin{equation}\label{estCalE}
|\calE(s,k)|\, \chi\left(\frac{ s}{\sqrt{k}\,a_k}\right) \leq \frac{C_3}{\sqrt{k}}
\end{equation}
for all $s\in\bbR$ and for all $k\in\bbN$.

	Substituting back into $I_2$, we get that $I_2 = J_1 + J_2$ where
	\[
	J_1:= \frac{e^{[Q_A(\eta)-|\eta|^2]/2}}{2\pi\sqrt{k}}\int_{-\infty}^{\infty}
	e^{-(Q_A(w)+1)s^2/2 + is\, \eta\, A w^T}\,
	\chi\left(\frac{ s}{\sqrt{k}\,a_k}\right)\, ds
	\]
	and
	\[
	J_2:= \frac{e^{[Q_A(\eta)-|\eta|^2]/2}}{2\pi\sqrt{k}}\int_{-\infty}^{\infty}
	e^{-(Q_A(w)+1)s^2/2 + is\, \eta\, A w^T}\,
	\calE(s,k)\,\chi\left(\frac{ s}{\sqrt{k}\,a_k}\right)\, ds.
	\]
	We now use the classic estimate
$
	\frac{1}{\sqrt{\pi}}\int_{-x}^x e^{-s^2}\, ds =  1 - \frac{e^{-x^2}}{x\sqrt{\pi}} + 
	O(\frac{e^{-x^2}}{x^2} )
$
to conclude that 
	\[
	J_1 = \frac{e^{[Q_A(\eta)-|\eta|^2]/2}}{2\pi\sqrt{k}}\int_{-\infty}^{\infty}
	e^{-(Q_A(w)+1)s^2/2 + is\, \eta\, A w^T}\,ds + O(1/k^{ 1/C}\,\log(k)^{1/2})
	\]
	and, using (\ref{estCalE}), that
$|J_2|\leq \frac{D}{k}$
	where $D$ is a constant that depends on $\eta$.
Given that $C<1$ we can conclude that
\begin{equation}\label{estimateI_2}
I_2 = \frac{1}{2\pi\sqrt{k}}e^{-|\eta|^2/2}  \int_{-\infty}^{\infty} e^{Q_A(isw + \eta)/2} \, e^{-s^2/2}\,ds + O(1/{k}).
\end{equation}
In view of (\ref{orderI3})
\[
\Upsilon_A(\eta, k) = \frac{1}{2\pi\sqrt{k}}e^{-|\eta|^2/2}  \int_{-\infty}^{\infty} e^{Q_A(isw + \eta)/2} \, e^{-s^2/2}\,ds + O(1/{k}),
\]
and the proof is complete.

\hfill{$\Box$}

\subsection{Inner product estimates}
In this section we prove (\ref{innerEstimate}), namely:

{\em Let $A, B \in \mathcal{D}_N, w \in S^{2N-1}$ and $\eta   \in \calH_w$, then
	\begin{equation*}\tag{\ref{innerEstimate}}
	\inner{\Psi_{A, w}}{\Psi_{B, w}}_{\calBk_{\bbC\bbP^{N-1}}}
	= \frac{2\pi}{k^N}\int_{\calH_w}\sigma_A(\eta)\,\overline{\sigma_B}(\eta) \, dL(\eta) + O(k^{-N-1}).
	\end{equation*}
}
\begin{proof}
By equivariance, without loss of generality we can take $w = (1,\vec{0})$.  We
introduce  a standard parametrization of a dense open set  $\calU\in\bbC\bbP^{N-1}$, 
containing the point $\varpi = \pi(w)$, namely, the
set $\calU$ which is the complement to the hyperplane  $\{z_1= 0\}$.  One identifies
$\calU\cong\bbC^{N-1}_\zeta$ by the coordinates
\begin{equation}\label{theCoordinates}
\zeta_j = \frac{z_{j+1}}{z_1}, \quad j=1,\ldots, N-1.
\end{equation}
Define next a section of $\pi:S^{2N-1}\to\bbC\bbP^{N-1}$ over $\calU$ by 
\begin{equation}\label{theSection}
S_\varpi: \bbC^{N-1}\to S^3,\quad S_\varpi(\zeta) = \frac{1}{\sqrt{1+|\zeta|^2}}\, (1, \zeta).
\end{equation}
Note that $\varpi$ corresponds to the origin $\zeta =0$, and $S_\varpi(0)= w$.


	The left-hand side of (\ref{innerEstimate}) 
	is an integral over $S^{2N-1}$ of a function that is $S^1$ invariant.
	Therefore, we can compute it (up to a factor of $2\pi$) by pulling it back by the section
	$S_\varpi$ and integrating with respect to the appropriate measure on $\bbC^{N-1}$.  
	A calculation 
	shows that 	
\[
	\inner{\Psi_{A, w}}{\Psi_{B, w}}_{\calBk_{\bbC\bbP^{N-1}}}
	=2\pi\int_{\bbC^{N-1}} \Psi_{A, w}(S_\varpi (\zeta)) \, \overline{\Psi_{B, w}} (S_\varpi (\zeta))  \frac{dL(\zeta)}{(1+|\zeta|^2)^N}  = \text{I + II}
\]
	where
$
	\text{I} = \int_{|\zeta|\leq 1} \Psi_{A, w}(S_\varpi (\zeta)) \, \overline{\Psi_{B, w}} (S_\varpi (\zeta))  \frac{dL(\zeta)}{(1+|\zeta|^2)^N}
$
	and II is the integral of the same integrand over $|\zeta|>1$.

We will show that II is rapidly decreasing.
	We first find a bound for $ \left| \Psi_{A, w}(S_\varpi (\zeta)) \right|$.
	To begin with, 
	\begin{align*}
	\left|\Psi_{A, w}(S_\varpi(\zeta)) \right| 
	&\leq \frac{1}{2 \pi} \int_0^{2 \pi} \left|\psi_{A,w}(e^{it} S_\varpi(\zeta)) \right| \, dt \\
	& = \frac{1}{2 \pi} \int_0^{2 \pi} \left|e^{kQ_A(e^{it} S_\varpi(\zeta)-w)/2} e^{-k|e^{it} S_\varpi(\zeta)-w|^2/2} e^{ik \omega (e^{it}S_\varpi(\zeta),w)/2}\right|  \, dt  \\
	&= \frac{1}{2 \pi} \int_0^{2 \pi} e^{k \Re [Q_A(e^{it} S_\varpi(\zeta)-w)/2]} e^{-k|e^{it}S_\varpi(\zeta)-w|^2/2}  \, dt \\
	&=  \frac{1}{2 \pi} \int_0^{2 \pi} e^{-k  \widetilde Q_A(e^{it} S_\varpi(\zeta)-w)/2} \, dt 
	\end{align*}
	where  $\widetilde{Q}_A(z) :=  -\Re (Q_A(z))+ |z|^2$ is a real positive definite quadratic form. 
	Denote by $c_A >0 $ the smallest eigenvalue of $Q_A$.  Then $\forall z$, $\widetilde{Q}_A(z) \geq c_A |z|^2$.
%
%
Hence 
	\begin{align*}
	\left|\Psi_{A, w}(S_\varpi(\zeta)) \right| 
	&\leq \frac{1}{2 \pi} \int_0^{2 \pi} e^{-kc_A|e^{it} S_\varpi(\zeta)-w|^2/2} \, dt  
	=  \frac{1}{2 \pi} \int_0^{2 \pi} e^{-kc_A| S_\varpi(\zeta)-e^{-it}w|^2/2} \, dt  \\
	&\leq \max _{t\in[0,2\pi]} e^{-kc_A| S_\varpi(\zeta)-e^{-it}w|^2/2} 
	= e^{-kc_A \min_{t\in[0,2\pi]}| S_\varpi(\zeta)-e^{-it}w|^2/2 }
	= e^{-kc_A (1-\rho(\zeta))}.
	\end{align*}
	where $\rho(\zeta) = (1+|\zeta|^2)^{-1/2}$. This last step results from the fact that
	\begin{equation*}
	|	S_\varpi(\zeta) - e^{-it}w|^2 = |(\rho-e^{-it}, \rho \zeta)|^2 = |\rho-e^{-it}|^2 + \rho^2 |\zeta|^2
	\end{equation*}
	which is minimized at $t=0$, and 
$ |\rho-1|^2 + \rho^2 |\zeta|^2= \rho^2(1+|\zeta|^2) +1- 2\rho = 1+1-2\rho = 2(1-\rho)$.

	All in all, we have $\left|\Psi_{A, w}(S_\varpi(\zeta)) \right|  \leq e^{-kc_A (1-\rho(\zeta))}$ and by similar analysis, we obtain \\ $\left| \Psi_{B, w}(S_\varpi(\zeta)) \right|  \leq e^{-kc_B (1-\rho(\zeta))}$ for some $c_B >0$. Therefore, 
	\begin{equation*}
	|\text{II}| \leq  2 \pi \int_{|\zeta| > 1} \left| \Psi_{A, w} (S_\varpi (\zeta)) \right| \;  \left| { \Psi_{B, w} } (S_\varpi (\zeta)) \right|  \frac{dL(\zeta)}{(1+|\zeta|^2)^N} \leq 2 \pi  \int_{|\zeta| > 1} e^{-k(c (1-\rho(\zeta))} \frac{dL(\zeta)}{(1+|\zeta|^2)^N}
	\end{equation*}
	where $c:= c_A + c_B$. If we change to polar coordinates, then $r = |\zeta|$ and $1-\rho(\zeta) = 1-(1+r^2)^{-1/2}$, so 
	\begin{align*}
	|\text{II} |\leq 2\pi \cdot (2 \pi)^{N-1} \int_{r=1}^\infty e^{-kc\left( 1 - \frac{1}{\sqrt{1+r^2}}\right)} \frac{r^{2N-3} \, dr}{(1+r^2)^N} 
	\leq C e^{-kc\left(1-\frac{1}{\sqrt{2}}\right)} 
	\end{align*}
 where $C >0$, an thus II tends to zero rapidly as $k \to \infty$.
	
	\medskip
	Now let's consider the integral I. 
	We change variables to $\zeta = \eta/\sqrt{k}$, so that $|\eta| \leq \sqrt{k}$ provided $|\zeta| <1$. Thus, 
	\begin{align*}
		|\text{I}|
		&\leq \frac{2\pi}{k^{N-1}} \int_{\C^{N-1}}  \left| \Psi_{A, w} (S_\varpi(\eta/\sqrt{k})) \right|  \left| \Psi_{B, w}(S_\varpi(\eta/\sqrt{k})) \right|  \chi(|\eta|/\sqrt{k}) \frac{dL(\eta)}{(1+|\eta|^2/k)^N} \\
		&= \frac{2\pi}{k^{N-1}} \int_{ \C^{N-1}}  \left| \Upsilon_A(\eta,k) \right|  \left| \Upsilon_B(n,k)\right|  \chi(|\eta|/\sqrt{k}) \frac{dL(\eta)}{(1+|\eta|^2/k)^N} 
	\end{align*}
	where $\chi$ is a cutoff function. We define 
	\begin{equation*}
	f_k(\eta) := \left| \Upsilon_A(\eta,k) \right|  \left| \Upsilon_B(n,k)\right|   \frac{\chi(|\eta|/\sqrt{k})}{(1+|\eta|^2/k)^N}. 
	\end{equation*}
	Now $f_k(\eta)>0$ is a sequence in $L^1(\C^{N-1}, dL)$ and $\exists \: c,C >0$ such that $f_k(\eta)$ is dominated by $Ce^{-c|\eta|^2}$, $\forall k, \eta$ such that $|\eta| \leq \sqrt{k}$. Moreover, $f_k(\eta)$ converges to $|\Upsilon_A(\eta,k)||\Upsilon_B(\eta,k)|$ pointwise as $k\to \infty$, so by the Dominated Convergence Theorem and by part 2 of Theorem 1.3,
	\begin{equation*}
	\inner{\Psi_{A, w}}{\Psi_{B, w}}_{\calBk_{\bbC\bbP^{N-1}}}
	= \frac{2\pi}{k^{N}} \int_{\bbC^{N-1}}\sigma_A(\eta)\,\overline{\sigma_B}(\eta) \, dL(\eta) + O(k^{-N-1}).
	\end{equation*}
	The additional factor of $1/k$ comes from the definition $\lim\limits_{k\to \infty} \Upsilon_A(\eta,k)= \sigma_A ({\eta}) /\sqrt{k}$, and similarly for $\Upsilon_B(\eta, k)$.
\end{proof}
Note: In the case where $B=A$, we have the norm of the reduced state in $\calBk_{\bbC\bbP^{N-1}}$ in terms of the $L^2-$norm of its symbol:
\begin{equation*}
\Vert \Psi_{A,w} \Vert_{\calB^{(k)}_{\bbC \bbP^{N-1}}}^2 
= \frac{2 \pi}{k^N} \int_{  \C^{N-1}} |\sigma_A(\eta)|^2 dL(\eta) + O(k^{-N-1}).
\end{equation*}

\begin{corollary}
	If $A,B \in \calD_N$ are such that $\sigma_A = \sigma_B$, then 
	\begin{equation*}
	\Vert  \Psi_{A,w} -  \Psi_{B,w}  \Vert^2_{\calBk_{\bbC\bbP^{N-1}}} =  O(k^{-N-1}).
	\end{equation*}
\end{corollary}	

\begin{proof}
	Applying the polarization identity and the previous result,
	\begin{align*}
	\Vert   \Psi_{A,w}  -  \Psi_{B,w}  \Vert^2_{\calBk_{\bbC\bbP^{N-1}}}
	&= \Vert  \Psi_{A,w}  \Vert^2_{\calBk_{\bbC\bbP^{N-1}}} + \Vert  \Psi_{B,w}  \Vert^2_{\calBk_{\bbC\bbP^{N-1}}} 
	- 2 \Re \inner{ \Psi_{A,w} }{ \Psi_{B,w}}_{\calBk_{\bbC\bbP^{N-1}}} \\
	&= \frac{2 \pi}{k^N} \int_{\C^{N-1}} \left(  |\sigma_A(\eta)|^2 +  |\sigma_A(\eta)|^2 - 2 \Re (\sigma_A(\eta) \overline{\sigma_B}(\eta) ) \right)  dL(\eta) + O(k^{-N-1}) \\
	&= \frac{2\pi}{k^N} \int_{\C^{N-1}} \left| \sigma_A(\eta) - \sigma_B(\eta)\right|^2 dL(\eta) + O(k^{-N-1}) \\
	&= O(k^{-N-1})
	\end{align*}
	since $ \sigma_A = \sigma_B$. 
\end{proof}

\medskip\noindent
{\em Proof of Proposition \ref{n2-case}.}  The second equality in (\ref{ketMu}) is a straightforward calculation, starting with 
(\ref{exactRedState}).   If $\mu = \rho_{(1,0)}(A)$, then $\ket{0,\mu}$ and $\Psi_{A,w}$ have the same symbol, and
the proposition follows from the previous corollary.

\subsection{The geometry behind the definition of the symbol}\label{subsec:GeomSymb}
The goal of this section is to discuss the notion of symbol of a coherent state in a 
general context of K\"ahler quantization.  
It is not logically needed in the proofs of our main results, but (hopefully) it sheds 
some light on the meaning of the symbol.

Intuitively, the symbol captures the asymptotic behavior of the state in a neighborhood of size $O(1/\sqrt{k})$ of its center.
It therefore interpolates between the behavior described by Theorem \ref{thm:stationary_proj} 
and part (1) of Theorem \ref{Main}.  
As a mathematical object, the symbol is a Schwartz function on the tangent space at the center of the state.  
Roughly speaking 
it arises by performing the rescaling $z=w+\frac{\eta}{\sqrt{k}}$ in suitable coordinates, where $w$ is
the center of the state, and taking the leading term as $k\to\infty$.
The result is a function of $\eta$.
An example is
of course (\ref{whatTheSymbolIs}), where it is crucial that $\eta$ is in the horizontal subspace $\calH_w$.

\subsubsection{Generalities on quantized K\"ahler manifolds}


Recall that a K\"ahler manifold $M$ 
is a complex manifold with a symplectic form $\omega$ which is
of type $(1,1)$, and such that the symmetric tensor
\[
g(u,v) := \omega(u, J(v)),\qquad J:TM\to TM\ \text{the complex structure}
\]
is positive definite.  

We begin by quoting the following theorem (see \S 7 of Chapter 0 in \cite{GH}):
\begin{theorem}
If $M$ is a K\"ahler manifold and $w\in M$, there exists a holomorphic coordinate system
$(z_1,\ldots ,z_N)$ centered at $w$ and such that the symplectic form near $w$ satisfies
\begin{equation}\label{}
\omega = i\sum_j dz_j \wedge d\zbar_j + O([2])
\end{equation}
where $O([2])$ designates a form whose components vanish quadratically at $w$.
\end{theorem}
We will say that such a coordinate system is {\em adapted} to $w$.

Let us now introduce $\calL \to M$ a Hermitian holomorphic line bundle with connection $\nabla$
with curvature the symplectic form $\omega$. (The precise meaning of this will be recalled soon.)
Let $\calU\subset M$ be an open set, and $s: \calU \to \calL|_\calU$ a local trivialization
of constant length equal to one.  Using $s$ we identify sections of $\calL|_\calU$ with
$C^\infty(\calU, \bbC)$.  If we let $\alpha\in \Omega^1(\calU)$ be the one-form on $\calU$
such that
\[
\nabla s = -is\otimes\alpha,
\]
then $\alpha$ is real-valued and we can identify
$\nabla = d -i\alpha$.
The precise relationship that we assume between the connection and the symplectic form
is that
\begin{equation}\label{}
d\alpha = \omega|_\calU.
\end{equation}

\begin{lemma}
For any $w\in M$ and any holomorphic coordinate system $(z_1,\ldots , z_N)$ adapted to $w$,
there exists a local unitary trivialization near
$w$ such that the corresponding connection form $\alpha$  satisfies
\begin{equation}\label{}
\alpha = \frac{i}{2}\sum_j z_j\, d\zbar_j - \zbar_j\, dz_j +O([2]).
\end{equation}
\end{lemma}
We will say that such a trivialization is {\em adapted} to $w$.  To our knowledge this
notion was introduced in \cite{BSZ}, in a more general context, under the name ``preferred frame".
\begin{proof}
Starting with any trivialization $s$, any other unitary trivialization is of the form
$t = e^{-if}s$ where $f$ is a smooth real-valued function on $\calU$.  Since
\[
\nabla t = -ie^{-if}s\otimes df+ e^{-if}\nabla s = -ie^{-if} s\otimes\left(df + \alpha\right),
\]
the connection form associated with $t$ is $\beta= \alpha + df$.  We will choose $f$
appropriately.  First, we choose $f$ so that $df_w = -\alpha_w$, which ensures that
the connection form associated with $t$ vanishes at $w$. 
Next, introduce holomorphic coordinates adapted to $w$, $(z_1,\ldots, z_N)$, and write their real and imaginary
parts as 
$z_j = \frac{1}{\sqrt{2}}\left(x_j -iy_j\right)$ .
Then $\omega_w = \sum_j dy_j\wedge dx_j|_w$.

Let us write
$(u_1,\ldots, u_{2N}) = (x_1, \dots , x_{N}, y_1,\ldots , y_N)$,
$\alpha = \sum_j \alpha_j du_j$ and $ A= (A_{ij})= \begin{pmatrix}
\frac{\partial \alpha_i}{\partial u_j}(0)
\end{pmatrix}$.  Note that the condition $d\alpha = \omega$ implies
\[
A-A^T = 
\begin{pmatrix}
0 & I \\
-I & 0
\end{pmatrix} =:J.
\]
Choose the second derivatives of $f$ at zero
to be 
\[
\begin{pmatrix}
\frac{\partial^2 f}{\partial u_i\partial u_j}(0)
\end{pmatrix}
= -\frac 12\left(A+A^T\right),
\]
and let $\beta = \alpha + df = \sum_j \beta_j du_j$.  Then
\[
\begin{pmatrix}
\frac{\partial \beta_i}{\partial u_j}(0)
\end{pmatrix}
= A -\frac 12\left(A+A^T\right) = \frac 12\left(A - A^T\right) = \frac 12\,J.
\]
These conditions determine the first and second derivatives of $f$ at the origin, and 
\[
\beta \equiv \frac 12\sum_{j=1}^N y_j\,dx_j - x_j\,dy_j = 
\frac{i}{2}\sum_j z_j\, d\zbar_j - \zbar_j\, dz_j
\]
modulo a one-form whose coefficients vanish quadratically at $w$.
\end{proof}

\subsubsection{Definition of symbols}

As motivation for the general definition, let us begin with an example and investigate the notions of the previous section
for $M=\bbC^N$.
The bundle is trivial and the connection is given by the global form
$\frac{i}{2}\sum_j z_j\, d\zbar_j - \zbar_j\, dz_j$, where $(z_1,\ldots, z_N)$ are the ordinary coordinates.
Fix $w\in\bbC^N$.  Then  
$\eta_j := z_j - w_j$
are adapted coordinates.
 We claim that the trivialization
\begin{equation}\label{}
s_w(z):= e^{i\Im(z\wbar^T)}
\end{equation}
is adapted to $w$.  It is clearly unitary, and
a calculation shows that (see (\ref{connection}))
\[
\nabla s_w = \frac 12\left( (\wbar-\zbar)dz + (z-w)d\zbar\right)\, s_w,
\]
so the connection form associated with $s_w$ is exactly
\[
\frac i2 \sum \eta_j d\overline{\eta}_j - \overline{\eta}_j d\eta_j.
\]
In terms of this section, a Gaussian coherent state centered at $w$ is of the form
\begin{equation}\label{}
\psi_{A,w}(z) = e^{kQ_A(z-w)/2}\, e^{-k|z-w|^2/2}\, s_w(z).
\end{equation}
In adapted coordinates
\[
\frac{\psi_{A,w}}{s_w^k}(\eta) = e^{kQ_A(\eta)/2}\, e^{-k|\eta|^2/2}
\]
Note that rescaling $\eta$ by $1/\sqrt{k}$ results in a $k$-independent function
(in general we will have to take the limit as $k\to\infty$).
We now define:
\begin{definition}\label{def:GaussianSymbol}
The symbol of the coherent state $\psi_{A,w}$ is the function of $\eta\in\bbC^N$
\begin{equation}\label{symbUpstairs}
\sigma_{\psi_{A,w}}(\eta) = \frac{\psi_{A,w}}{s_w} \left( \frac{1}{\sqrt{k}}\eta\right) = e^{Q_A(\eta)/2 - |\eta|^2/2}.
\end{equation}
\end{definition}

Back to the general context,
let $\psi^{(k)}\in C^\infty(M, \calL^k)$ be a sequence of holomorphic sections, pick $w\in M$ and choose
holomorphic coordinates on an open set $\calU$
adapted to $w$ as well as an adapted trivialization $s_w$.  On $\calU$ we can
write
\begin{equation}\label{}
\frac{\psi^{(k)}}{s_w^k}\in C^\infty(\calU,\bbC).
\end{equation}
If $\eta$ denotes the adapted coordinates, and if $\psi$ is a coherent state 
with center at $w$, one can define its symbol as the function 
of $\eta$, if it exists, given by the leading asymptotics as $k\to\infty$ of
\begin{equation}\label{}
\Upsilon(\eta, k):= \frac{\psi^{(k)}}{s_w^k} \left( \frac{1}{\sqrt{k}}\eta\right).
\end{equation}
Below we will check that the symbol of the reduced states $\Psi_{A,w}$ is exactly obtained
in this way.  As another example, a general K\"ahler manifold carries a family of ``non squeezed"
coherent states, \cite{R}, whose symbols are the Gaussians $e^{-\norm{\eta}^2/2}$ where the
norm is the Riemannian metric, see Theorem 3.2 in \cite{BSZ}.

The definition of $\Upsilon(\eta,k)$ depends on the choices of adapted coordinates and 
trivialization.  To what extent does the leading asymptotics as $k\to\infty$ depend on 
these choices? 
If $t_w$ is another adapted trivialization, then $t_w = e^{if}s_w$ where the first and
second derivatives of $f$ vanish at $w$.  Therefore, in a given adapted coordinate system,
\[
\frac{\psi^{(k)}}{s_w^k} \left( \frac{1}{\sqrt{k}}\eta\right)= 
e^{ikf(0) + O(1/\sqrt{k})}\frac{\psi^{(k)}}{t_w^k} \left( \frac{1}{\sqrt{k}}\eta\right).
\]
Thus the ambiguity inherent in the choice of adapted section translates, asymptotically, into an overall
oscillatory factor $e^{ikf(0)}$.  (This is in agreement with the fact that the center of a coherent
state really is a point in the pre-quantum circle bundle.)  

\medskip
We show next that the general procedure outlined above agrees with the definition (\ref{whatTheSymbolIs}) of the
symbol of a reduced state.

\subsubsection{Symbols of reduced Gaussian states}

We need to explain how the previous discussion corresponds to part (3) of Theorem \ref{Main}.
We need first to clarify the way in which homogeneous polynomials
of degree $k$ on $\bbC^N$ can be seen as sections of $\LCP^k\to\bbC\bbP^{N-1}$.

As a space, $\LCP$ is the quotient of $S^{2N-1}\times\bbC$ by the equivalence relation $\sim_k$
defined as
$(e^{i\theta}z\,,\, \lambda) \sim_k (z, e^{-ik\theta}\lambda)$.
If $\psi: S^{2N-1}\to\bbC$ is any function such that 
\begin{equation}\label{homogen}
\forall e^{i\theta}\in S^1,\ z\in S^{2N-1}\qquad \psi(e^{i\theta}z) = e^{ik\theta}\psi(z),
\end{equation}
we can associate to it a section $s_\psi: \bbC\bbP^{N-1}\to\LCP^k$ by:
\begin{equation}\label{sPsi}
\bbC\bbP^{N-1} \ni \pi(w) \xrightarrow{s_\psi} [(w,\psi(w))]_k \in\LCP^k
\end{equation}
where $\pi: S^{2N-1}\to\bbC\bbP^{N-1}$ is the projection, and $[(w,\psi(w))]_k$ is the 
$\sim_k$ equivalence class of $(w,\psi(w))$.  One can easily check that (\ref{sPsi}) is well-defined:
 $\pi(w) = \pi(w')\ \Rightleftarrow\ w' = e^{i\theta}w$ for some $e^{i\theta}$, and so
 \[
 (w', \psi(w')) = (e^{i\theta}w, \psi(e^{i\theta}w)) = (e^{i\theta}w, e^{ik\theta}\psi(w)) \sim_k (w, \psi(w)).
 \] 
 One can also check that, conversely, any section of $\LCP$ is an $s_\psi$ for some
 $\psi$ as above.  The holomorphic sections correspond to $\psi$'s that are restrictions of
 holomorphic functions, and the homogeneity condition implies that they must be polynomial
 functions.

Let us return to the problem of computing the symbols of reduced Gaussian states, in the sense of
this section.
For simplicity of notation we only discuss the $N=2$ case in detail.  

Given the covariance of the construction of reduced Gaussian states $\psi_{A,w}$ with respect to the action 
of the unitary group, it suffices to analyze a particular choice of $w$.  
We take again take $w=(1,0)$, and introduce the coordinates (\ref{theCoordinates}) and the section
$S_\varpi$ given by (\ref{theSection}).
%
The latter induces unitary trivializations of all $\LCP^k$ over $\calU$, by
\begin{equation}\label{}
\calU\ni \zeta\xrightarrow{s_\varpi^k} [(S_\varpi(\zeta), 1)]_k.
\end{equation}

\begin{lemma}
	The coordinate $\zeta$ and the section $s_\varpi$ are adapted to $\zeta$.  Moreover, for any
	section $s_\psi: \calU\to\LCP^k$ where $\psi: S^{2N-1}\to\bbC$ satisfies (\ref{homogen}),
	\begin{equation}\label{expressionSection}
	\frac{s_\psi}{s_\varpi^k} = \psi\circ S_\varpi.
	\end{equation}
\end{lemma}
\begin{proof}
	The connection form associated with $s_\varpi$ is  the pull back of the connection form \\ $\alpha = \frac i2\left( zd\zbar -\zbar dz\right)$ by $S_\varpi$.  A calculation shows that this
	equals
\begin{equation}\label{}
	S_\varpi^*\alpha =\frac i2\,\frac{1}{1+|\zeta|^2}\, \left( \zeta d\zetabar - \zetabar d\zeta\right),
\end{equation}
	which is clearly of the form $\frac i2 \left( \zeta d\zetabar - \zetabar d\zeta\right) + O(|\zeta|^2)$.  Differentiating the
	above, after further computations we obtain the expression for the reduced symplectic form
	\begin{equation}\label{}
	\omega_{\bbC\bbP^1} = \frac{i}{(1+|\zeta|^2)^2}\, d\zeta\wedge d\zetabar.
	\end{equation}
	
	The second statement follows from the definition of $s_\varpi$ and the relationship between $s_\psi$ and $\psi$,
	(\ref{sPsi}).
\end{proof}

According to the general procedure for the computation of the symbol of a coherent state, and taking into account
(\ref{expressionSection}), we need to compute the asymptotics of
\begin{equation}\label{upsilon}
\Psi_{A,(1,0)}\left[S_\varpi\left(\frac{\eta}{\sqrt{k}}\right)\right].
\end{equation}
A Taylor expansion of $S_\varpi$ at the origin easily gives that
\[
S_\varpi\left(\frac{\eta}{\sqrt{k}}\right) = w + \frac{\eta}{\sqrt{k}} + O(1/k),
\]
so the asymptotics of (\ref{upsilon}) agrees with that of $\Upsilon_A(\eta, k)$
to leading order.

\section{Reduction of symbols and the metaplectic representation}

This material will be used to prove the propagation theorem of \S 5.
Throughout this section $\h = k=1$, and we work entirely in the category of symplectic vector spaces.

This section has two goals.  First, we interpret the passage from the
symbol of $\psi_{A,w}$ to the symbol of its reduction $\Psi_{A,w}$
as applying an operator of reduction in the Heisenberg representation
(Lemma \ref{lemma:symbIsReduction}).
Second, we show that the metaplectic representation of certain linear symplectomorphisms
is covariant with respect to the above procedure (Proposition \ref{prop:MpIsNat}).

\subsection{Bargmann spaces}
Let $(E,\omega, J)$ be a K\"ahler vector space.  We take the sign convention that the 
associated positive definite metric is $g(u,v) = \omega(u, Jv)$.  
A nice reference for the material in this section is \cite{D}.  We will quote freely
from that article.

Let
\begin{equation}\label{}
\calB(E)= \left\{ \psi:E\to\bbC\;;\; \psi(v) = f(v)e^{-\norm{v}^2/2}\text{ where }
\dbar f = 0\text{ and } \psi\in L^2(E)
\right\}
\end{equation}
be the Bargmann space of $E$.  
Here $\delbar$ is the d-bar operator associated with $J$ and $\norm{v}^2 = g(v,v)$.

In our applications $E$ is the tangent space at a point in a K\"ahler manifold.  
The symbols of squeezed states at that point will be elements of $\calB(E)$.

The Heisenberg group of $E$ is unitarily represented in $\calB(E)$, as follows.
If $a\in E$, define the operator $\rho(a): \calB\to\calB$ by
$
\rho(a)(\psi)(v) = e^{i\omega(a,v)}\psi(v-a).
$
Then $\rho(a)\circ\rho(b) = e^{i\omega(a,b)}\rho(a+b)$, so these operators form part
of the Heisenberg  representation of the Heisenberg group of $E$.  
Recall that $\psi\in\calB$ is said to be a smooth vector 
iff for all $\phi\in\calB$ the function
$a\mapsto \inner{\rho(a)(\psi)}{\phi}$ is smooth (this is the analogue of Schwartz
functions in Bargmann space).
We will
denote by
\[
\calB^\infty(E) \subset \calB(E)
\]
the subspace of smooth vectors for this representation.

\subsection{Reduction}
If $S\subset E$ is a subspace, we denote by $S^\circ$ and $S^\bot$ its 
symplectic annihilator and orthogonal complement, respectively.  Note that
\begin{equation}\label{noteThat}
J(S^\circ) = S^\perp.
\end{equation}

\medskip
From now on we fix a co-isotropic subspace $\calC\subset E$ (this means that
$\calC^\circ \subset\calC$).  Let us define
\begin{equation}\label{}
\calH:= \calC\cap J(\calC),\  \text{the maximal complex subspace of } \calC.
\end{equation}
Note that automatically $\calH$ is a K\"ahler (in particular, symplectic) subspace of $E$.

\begin{lemma}
	One has:
	\begin{equation}\label{maxCpxSub}
	\calC\cap (\calC^0)^\bot = \calH,
	\end{equation}
	and therefore the projection $\pi: \calC\to \calC/\calC^\circ =: F $ identifies the reduction,
	$F$, of $\calC$ with the maximal complex subspace of $\calC$.  Under this identification
	the symplectic structures of $\calH$ and $F$ agree.
\end{lemma}
\begin{proof}
	By (\ref{noteThat}), $J(\calC) = \left(\calC^\circ\right)^\bot$ and (\ref{maxCpxSub}) follows,
	which implies that $\pi$ restricted to $\calH$ is a bijection.
	The rest follows from the usual characterization of the symplectic structure of a reduction.
\end{proof}


By the previous discussion, the reduction $F = \calC/\calC^\circ$ of $\calC$ inherits the
structure of a K\"ahler vector space.  Let $\calB(F)$ denote its Bargmann space, 
and $\calB_F^\infty\subset\calB_\calF$ the subspace of smooth vectors.  Our objective
is to introduce a natural ``reduction" operator
\begin{equation}\label{}
\calR: \calB^\infty(E) \to \calB^\infty(F)
\end{equation}
associated with $\calC$.
Here ``natural" is with respect to symplectic linear transformations.  There is 
an obvious map, namely restriction to $\calH$ followed by the identification $\calH\cong F$,
but this is not the right one for our purposes.
\begin{definition}
	We define $R: \calB^\infty(E) \to \calB^\infty(F)$ to be the operator of
	restriction to $\calC$ followed by integration over $\calC^\circ$, 
	with respect to the measure induced by the Euclidean inner product.
\end{definition}

The point of this definition is that it describes the abstract way to
construct the symbol of a reduced state:

\begin{lemma}\label{lemma:symbIsReduction}
In the context of Theorem \ref{Main}, one has:
\[
\sigma_A(\eta) = \frac{1}{2\pi}R(\sigma_{\psi_{A,w}})(\eta),
\]
where $E = \bbC^N$, $\calC = T_w S^{2N-1}$ and
$\sigma_{\psi_{A,w}}(z)=e^{Q_A(z)/2-|z|^2/2}$ is the symbol of $\psi_{A,w}$.
\end{lemma}
\begin{proof}
Simply note that $iw$ is a unitary basis of $\calC^\circ$ and $\eta\in\calH_w$.
Therefore 
\[ e^{-|\eta|^2/2}  \int_{-\infty}^{\infty} e^{Q_A(isw + \eta)/2} \, e^{-s^2/2}\,ds
\] 
is exactly the definition of $R(\sigma_{\psi_{A,w}})(\eta)$.
\end{proof}

We now explicitly compute $\sigma_A$ in a special case:

\begin{lemma}\label{lemma:redSymbSpecialCase}
	If $w=(1,\vec{0})$, then 	(see \eqref{symbDownstairs})
		\begin{equation*}
	\sigma_A(\eta) =\frac{1}{\sqrt{2\pi}} \frac{1}{\sqrt{Q_A(w)+1}} \; e^{Q_{\rho_w(A)}(\eta)/2}\,
	e^{-|\eta|^2/2}
	\end{equation*}
	where $\rho_w(A)\in\calD_{N-1}$ is the lower $(N-1) \times (N-1)$ principal minor of 
	\[
	A- \frac{Aw^TwA}{wAw^T+1}.
	\]
\end{lemma}

\begin{proof}
%
	Since
	\begin{align*}
		Q_A(isw+\eta) 
		= i^2 s^2 wAw^T + 2is \eta A w^T + \eta A \eta ^T 
		= -s^2 Q_A(w) +2i s \eta A w^T + Q_A(\eta)
	\end{align*}
	equation \eqref{whatTheSymbolIs} may be re-written as
	\begin{align*}
	\sigma_{A} (\eta) 
		& = \frac{1}{2 \pi} e^{-|\eta|^2/2} \, e^{Q_A(\eta)/2}\int_{-\infty}^\infty e^{-s^2(Q_A(w)+1)/2}\, e^{s i \eta A w^T} \, ds \\
		&= \frac{1}{2 \pi} e^{-|\eta|^2/2} \, e^{Q_A(\eta)/2}\int_{-\infty}^\infty e^{-b^2s^2/2}\, e^{c s} \, ds
	\end{align*}
	where $b^2 := Q_A(w)+1$ with $b$ in the right side of the complex plane and $c := i \eta A w^T$. Now $\Re(b^2) = \Re(Q_A(w)+1)>0$ since $A \in \calD_N$ and $|w|=1$, so then we can evaluate the integral, 
	\begin{align*}
		\sigma_{A} (\eta) 
		&= \frac{1}{2 \pi} e^{-|\eta|^2/2} \, e^{Q_A(\eta)/2} \: \frac{\sqrt{2\pi}}{b} e^{(c/b)^2/2} 
		= \frac{1}{\sqrt{2 \pi}} \frac{1}{\sqrt{Q_A(w)+1}} e^{Q_A(\eta)/2} \, e^{-(\eta A w^T)^2/(2(Q_A(w)+1))} \, e^{-|\eta|^2/2}.
	\end{align*}
	Since $\eta \in \calH_{(1,\vec{0})}$, our choice of $w$ forces the first coordinate of $\eta$ to be zero, so we
	take $\eta = (0, \eta_1, \dots, \eta_{N-1})$ and write
	\begin{align*}
		\frac{1}{2}\left[Q_A(\eta) - \frac{(\eta A w^T)^2}{Q_A(w)+1}\right] 
		= \frac{1}{2} \eta \left[A - \frac{A w^TwA}{wAw^T+1}  \right] \eta^T.
	\end{align*}
	Therefore the matrix $\rho_w(A)$ is the lower $(N-1) \times (N-1)$ principal minor of the matrix 
	\begin{equation*}
		A - \frac{A w^TwA}{wAw^T+1}.
	\end{equation*}
\end{proof}

\begin{corollary}\label{prop:SameReduction}
	The symbol of the $\Psi_{A,w}$ for any $w\in S^{2N-1}$ is given by equation (\ref{symbDownstairs}):
	\begin{equation*}
		\sigma_A(\eta) =\frac{1}{\sqrt{2\pi}} \frac{1}{\sqrt{Q_A(w)+1}} \; e^{Q_{\rho_w(A)}(\eta)/2}\,
		e^{-|\eta|^2/2}
	\end{equation*}
	for a suitable $Q_{\rho_w(A)} \in \calD(\calH_w)$.
\end{corollary}
\begin{proof}
By equivariance of the construction under the action of U$(N)$, we can assume without loss of generality that $w=(1,\vec{0})$.
But that case was settled in Lemma \ref{lemma:redSymbSpecialCase}.
\end{proof}

\subsection{The metaplectic representation and reduction}
We now turn to the naturality of the reduction operator with respect to changes of the complex structure.  Once again, let $(E,\omega, J)$ be a K\"ahler vector space.  Denote by 
$P_E: L^2(E)\to \calB(E)$ the orthogonal projector (it turns out that $\calB(E)$ 
is closed in $L^2(E)$).
If $\Phi: E\to E$ is a symplectic transformation, then 
one can form the unitary operator
$U_\Phi: L^2(E)\to L^2(E)$ which is simply
\[
U_\Phi (\psi) = \psi\circ \Phi^{-1}.
\]
One of the main results of \cite{D} is the following:
\begin{theorem} (\cite{D} \S 6)  Let $\text{Sp}(E)$ denote the group of symplectic transformations of $E$.
The assignment
\[
\text{Sp}(E)\ni\Phi\mapsto \calW(\Phi):=  \eta_{J,\Phi}\,
 P_E\circ U_\Phi : \calB(E)\to \calB(E)
\]
where
\begin{equation}\label{}
\eta_{J,\Phi} = 2^{-N}\left(\det\left[(I-iJ)+\Phi(1+iJ)\right]\right)^{1/2}
\end{equation}
is the metaplectic representation.
\end{theorem}

\medskip
Our goal here is to prove that the metaplectic representation is natural with respect to symplectic quotients,
in the following sense.  
Let $\calC\subset E$  be a co-isotropic subspace, as above, and let $\Phi:E\to E$ a linear symplectic isomorphism
satisfying:
\begin{itemize}
	\item[(1)] $\Phi(\calC) = \calC$.
\end{itemize}
From this it follows that $\Phi$ maps $\calC^\circ$ onto itself.
Let us further assume that
\begin{itemize}
	\item[(2)] the restriction of $\Phi$ to $\calH^\circ = \calC^\circ + J(\calC^\circ)$ is the identity:  $\Phi|_{\calH^\circ}: \calH^\circ \to \calH^\circ$.
\end{itemize}
Denote by $F = \calC/\calC^\circ$ the reduction of $\calC$, and by 
$\phi: F\to F$ the reduction of $\Phi$:
\[
\forall v\in \calC\qquad \phi([v]) = [\Phi(v)],
\] 
where $[v]\in F$ denotes the projection of $v$.  $\phi$ itself is a symplectomorphism.  

\begin{proposition}\label{prop:MpIsNat}
	Under the previous assumptions (1) and (2), the following diagram commutes,
\begin{equation}\label{diagCommutes}
\begin{array}{rccc}
 &\calB(E) & \xrightarrow{\calW(\Phi)}& \calB(E)\\
 & \downarrow     &                               & \downarrow\\
 & \calB(F)	& \xrightarrow{\calW(\phi)}&   \calB(F)
\end{array}
\end{equation}
where the vertical arrows are the reduction operator $R$.
\end{proposition}
\begin{proof}
	Let $\calH$ be as in (\ref{maxCpxSub}).  Since $\calH$ is a symplectic subspace of $E$, one has that
	$E = \calH\oplus \calH^\circ$.  Moreover one identifies $\calH$ with $F$ as K\"ahler vector spaces.
	
	\medskip\noindent
	{ 1.}  We begin by showing that {\em  with respect to this decomposition
	$\Phi$ is of the form}
\[
\calM = \begin{pmatrix}
\phi & 0 \\
0 & I_{\calH^\circ}
\end{pmatrix}
\]
with $\phi\in\text{Sp}(\calH)$.
Already the assumptions on $\Phi$ imply that $\calM$ is of the form
\[
\calM = \begin{pmatrix}
\phi & 0 \\
B & I_{\calH^\circ}
\end{pmatrix}.
\]
Introduce now symplectic bases of $\calH$ and $\calH^\circ$, and replace $\calM$ by the
corresponding matrix.  Then the condition that $\calM$ is symplectic is that
$J = \calM J \calM^T$ where
\[
	J = 
	\begin{pmatrix}
	J_r & 0 \\
	0 & J_{N-r}
	\end{pmatrix}
\]
with  $J_r = \begin{pmatrix} 0 & -I_r \\ I_r & 0	\end{pmatrix}$ (and similarly for $J_{N-r}$).
This implies that $\phi J_r B^T =0$ and $\phi J_r \phi^T = J_r$, which in turn imply that $B=0$.

\medskip\noindent
{ 2.}  Next, consider the metaplectic representation of $\Phi$ in the Bargmann space $\calB(E)$ of $E$.  
The direct sum 
decomposition $E = \calH\oplus \calH^\circ$ implies that
\[
\calB(E) = \calB(\calH)\widehat{\otimes}\calB(\calH^\circ)
\]
(tensor product of Hilbert spaces), and the discussion above easily implies that 
$\text{Mp}(\Phi) = \text{Mp}(\phi)\otimes I$, where $\phi:\calH\to \calH$ is the restriction of $\Phi$ 
to $\calH$.

One can then easily check the commutativity of the diagram (\ref{diagCommutes}) on 
elements of $\calB(E)$ that are pure tensor products $\psi_1\otimes \psi_2$
(since the procedure of reduction is the identity on the first factor).  However, the
span of such elements is dense in $\calB(E)$, and therefore the diagram must commute on all elements
of $\calB(E)$.
\end{proof}

\section{Propagation}

In this section we investigate the semi-classical (or large $k$) limit of the quantum dynamics 
of our squeezed states.  

\subsection{Classical dynamics}
We begin with classical dynamics.  
Let $h:\bbC\bbP^{N-1}\to\bbR$ be a smooth Hamiltonian.  Let us define
\begin{equation}\label{upstairsHam}
H:\bbC^N\setminus\{0\}\to\bbR, \quad H(z) := |z|^2 h\left(\pi\left[\frac{z}{|z|}\right]\right),
\end{equation}
where, recall, $\pi: S^{2N-1}\to\bbC\bbP^{N-1}$ is the projection.  
We will call $H$ the {\em canonical lift} of $h$.

Clearly $H$ is positive-homogeneous of degree two and $S^1$ invariant, in the following sense:
\begin{equation}\label{propsH}
\forall \lambda\in\bbC^*, \  z\in \bbC^N\setminus\{0\}\qquad H(\lambda z) = |\lambda|^2 H(z).
\end{equation}
Conversely, any $H:\bbC^N\setminus\{0\}\to\bbR$ with this property is related to a smooth function
$h$ on $\bbC\bbP^{N-1}$ by (\ref{upstairsHam}).

The following is almost immediate:
\begin{lemma}\label{lemma:isomUpDownstairs}
	The trajectories of the Hamilton flow of the canonical lift of $h$ on $S^{2N-1}$ project onto trajectories
	of the Hamilton flow of $h$.
\end{lemma}

 will be useful below:
\begin{lemma}\label{lemma:dynamicalImplications}
Consider $h\in C^\infty(\bbC\bbP^{N-1}) $ and $\varpi\in\bbC\bbP^{N-1}$ a critical point.  Let $w\in\pi^{-1}(\varpi)$ and
$\calH\subset T_wS^{2N-1}$ be the horizontal space at $w$, which we identify with $T_\varpi\bbC\bbP^{N-1}$.

If, in addition, $h(\varpi)=0$ then 
$w$ is a critical point of the canonical lift, $H$, of $h$, and with respect to the decomposition
$T_w\bbC^N = \calH\oplus \calH^\circ$ the Hessian of $H$ at $w$ has the block form
\[
\begin{pmatrix}
\ast & 0 \\
0 & 0
\end{pmatrix}
\]
where $\ast$ is the  Hessian of $h$ at $\varpi$.
\end{lemma}
\begin{proof}
Let $G(z)= h\left(\pi\left[\frac{z}{|z|}\right]\right)$, so that $H(z) = |z|^2G(z)$.
Since $G(w)=0$, $dH_w = dG_w$.  Since $G$ is homogeneous of degree zero,
$dG(\nu_w) = 0$ where $\nu_w$ is the unit normal to the sphere at $w$.
It is also clear that $dG(\partial_\theta) = 0$, and since $\varpi$ is a critical
point of $h$, $dG_w$ is zero on horizontal vectors as well.  Therefore $dG_w=0$, and 
$w$ is a critical point of $H$.

Let us now consider the Hessian.  By the product rule for Hessians
\begin{alignat*}{1}
\hess{H}_w &= G(w)\hess{|z|^2}_w + |w|^2\hess{G}_w + d(|z|^2)_w\otimes dG_w + dG_w\otimes d{|z|^2}_w \\
\notag		&= \hess{G}_w,
\end{alignat*}
which implies the desired result.
\end{proof}

\medskip
\begin{remark}
If $\varpi$ is a critical point of $h$ but $h(\varpi)$ is not necessarily zero, then we can apply the
previous lemma to $\tilde{h} = h -h(\varpi)$.  Clearly the canonical lifts of these functions are related by
$\widetilde{H} = H - h(\varpi)|z|^2$, and since $\PB{H}{|z|^2}=0$ the Hamilton flow of $H$ restricted to 
the unit sphere agrees with that of $\widetilde{H}$ up to a the action of $e^{ith(\varpi)}\in S^1$.
\end{remark}

\subsection{Quantum propagation}

\subsubsection{Quantization of functions on $\bbC\bbP^{N-1}$}

Let $h:\bbC\bbP^{N-1}\to\bbC$ be a smooth Hamiltonian and $H$ its canonical lift, which we extend to a smooth function on
$\bbC^N$ cutting it off near zero by a radial function. The Weyl quantization of $H$, $\widehat H$, in Bargmann space commutes with the
quantized circle action, and we obtain operators
\[
\hat{h} : \calBk_{\bbC\bbP^{N-1}}\to \calBk_{\bbC\bbP^{N-1}}
\]
simply by restricting $\widehat H$ to $\calBk_{\bbC\bbP^{N-1}}$.  
{\em We will take the sequence of these operators to be the quantization of $h$.}
This recipe is not entirely well-defined due to the cutoff, but different choices of cutoffs lead to equivalent asymptotic
estimates.  It also agrees asymptotically with the Berezin-Toeplitz quantization of $h$.

Our first observation is:
\begin{proposition}
	One has:
	\[
	\hat{h} (\Psi_{A,w}) = h((\pi(w))\Psi_{A,w}\left(1 + O(1/\sqrt{k})\right).
	\]
\end{proposition}
\begin{proof}
	The analogous result for the action of $\wh{H}$ on Gaussian coherent states in $L^2(\bbR^N)$ 
	is well-known.   Since $[\wh{H}, \calR] = 0$, the result follows immediately.
\end{proof}

\subsubsection{Propagation of squeezed states in Bargmann space}

We begin by reviewing the propagation of Gaussian states in Bargmann space.    

We need to introduce some notation.  Let $H:\bbR^{2N}\to\bbR$ be a
smooth Hamiltonian which agrees with the canonical lift of a smooth $h:\bbC\bbP^{N-1}\to\bbR$
outside a small neighborhood of the origin, $\wh{H}:\calBk_{\bbC^{N}}\to\calBk_{\bbC^{N}}$ its Weyl
quantization in Bargmann space and $U(t)$ the fundamental solution of the Schr\"odinger equation
$i\h \partial_t U = \wh{H}U$. 
Let $w\in\bbC^N$, $t\mapsto w(t)$ be the trajectory of $H$ through $w$.
For each $t\in\bbR$, let
\begin{equation}\label{}
S(t) :=   H_{\zbar z}(w(t))
\quad\text{and}\quad
R(t) :=  \frac 12 H_{zz}(w(t))
\end{equation}
with $H_{\zbar z} = \begin{pmatrix}
\frac{\partial^2 H}{\partial\zbar_j z_l}
\end{pmatrix}$ etc.
Then, one has:

\begin{theorem}\label{thm:EuclideanProp} (c.f. \cite{CR} \S 4)  Let $A\in\calD_N$.  Then
	\begin{equation}\label{propaUpstairs}
	U(t)\left(\psi_{A,w}\right) = \nu(t) e^{ i k \delta_t}\psi_{A(t),w(t)}\left(1 + O(1/\sqrt{k})\right)
	\end{equation}  
	where:
	\begin{enumerate}
		\item $A(t)$ and $\nu(t)$ solve 
		\begin{align}
		\dot{A}  &= -2i \left(R + \frac{1}{2}(SA+A{S}^T) + A \bar{R} A\right) \quad\text{and} \label{Adot} \\
		\frac{\dot{\nu}}{\nu} &= -i  \left(\frac{\tr(S)}{2} +  \tr(\bar{R}A)\right),  \label{cdot}
		\end{align}
		with $A(0)=A$ and $\nu(0)=1$, and
		\item $\delta_t = -tH(w) + \frac i2\int_0^t  \left(w(s)\dot\wbar(s)-\dot w(s)\wbar(s)\right) ds$.
	\end{enumerate}
	The estimates are uniform for $t$ in a compact interval.
\end{theorem}
In \cite{CR} \S 4 the authors prove a more general result on the propagation of coherent states
in $L^2(\bbR^N)$.  The proof of the previous theorem follows exactly the same scheme.
For ease of reference we sketch the proof in the Appendix.

\begin{remarks} Let us look at some special cases.  
	\begin{enumerate}
		\item If $R\equiv 0$ and $S$ is time-independent, then the Hamilton flow of $\calQ$
		is a one-parameter group of unitary transformations.  The solutions to (\ref{Adot}) and (\ref{cdot}) are
		\[
		A(t) = e^{-itS}A e^{-itS^T}\quad\text{and}\quad \nu(t) = e^{-it \tr (S)/2}.
		\]
		Note that,
		by the covariance property (\ref{covariance}) $\psi_{A(t),0}$ is simply the rotation of $\psi_{A,0}$ by $e^{-itS}$.
		
		\item If, instead, $S\equiv 0$ we get a ``squeezing" effect.  The modulus of the 
		prefactor $\nu(t)$ adjusts the $L^2$ norm of $\psi$ so that it is constant in time.
	\end{enumerate}
	
\end{remarks}

It is instructive to note that 
\begin{equation}\label{}
\psi_{A(t),w(t)} = e^{ik \delta_t}\wh{T}_{w(t)} U_{\calQ_{(t) }}\wh{T}_w^{-1}(\psi_{A, w}),
\end{equation}
where $\wh{T}_w(f) (z) = e^{-k|w|^2/2} e^{kz\wbar} f(z-w)$ is the quantum translation by $w\in\bbC^N$
and $U_{\calQ(t)}$ is the propagator of the Weyl quantization of $\calQ(t)$, $1/2$ the 
quadratic form associated to the Hessian of $H$ at $w(t)$.
$U_{\calQ(t)}$ is the metaplectic operator associated with the Jacobian of $\phi_t:\bbC^N\to\bbC^N$ at $w$, 
where $\{\phi_t\}$ is the Hamilton flow of $H$ (defined by continuity from the identity at $t=0$):
\begin{equation}\label{mPExplicit}
\nu(t)\psi_{A(t), 0} = \text{Mp}(\text{Jac}(\phi_t)_w)(\psi_{A,0}).
\end{equation}
At the level of symbols, (\ref{symbUpstairs}), one can re-write this as
\begin{equation}\label{symbPropaUpstairs}
\nu(t)\sigma_{\psi_{A(t),w(t)}} = \text{Mp}(d(\phi_t)_w)\left(\sigma_{\psi_{A,w}}\right),
\end{equation}
{\em provided one identifies tangent spaces $T_w\bbC^N \cong T_{w(t)}\bbC^N$ using translations}.
This identification is natural, using the affine structure of Euclidean space.  In contrast, no
such identification exists among tangent spaces of $\bbC\bbP^{N-1}$, which complicates the 
description of the symbol of a propagated reduced state.

\subsubsection{Propagation of the reduced coherent states}
The propagation of reduced states follows easily from the Euclidean case.  
Let $h:\bbC\bbP^{N-1}\to \bbR$ be smooth.
We will denote by
\[
V (t) = e^{-ikt \hat{h} }: \calBk_{\bbC\bbP^{N-1}}\to \calBk_{\bbC\bbP^{N-1}}
\]
the quantum propagator on the Bargmann space of the projective space.  In this section we investigate the propagation
$V (t)(\Psi_{A,w})$ of reduced Gaussian states.
The first result is that, to leading order, the propagation
of a squeezed state remains a squeezed state.  
\begin{theorem}\label{thm:FirstPropaThm}
	The evolution $V (t)(\Psi_{A,w})$ of a reduced Gaussian state is of the form
	\begin{equation}\label{redSymbolPropa}
	V (t)(\Psi_{A,w}) = \nu(t) e^{ik\delta_t}\Psi_{A(t),w(t)}\,\left(1 + O(1/\sqrt{k})\right),
	\end{equation}
	where $w(t)$, $\nu(t)$ and $\delta_t$ are as in Theorem \ref{thm:EuclideanProp} with $H$ the canonical lift of $h$.
	The estimates are uniform for $t$ in a compact interval.
\end{theorem}
\begin{proof}
	Let $U (t) = \exp\left[-ikt\widehat{H}\right]$.  Simply notice that 
	$[U , \calR ] = 0$, as $\widehat{H}$ and the harmonic oscillator commute and $\calR $ is a
	normalized spectral projector of the latter, and apply Theorem \ref{thm:EuclideanProp}.
\end{proof}

\medskip
Next we address the problem of computing the symbol of $V (t)(\Psi_{A,w})$ for each $t$.
Recall that this symbol is an element of the $\h=1$ Bargmann space
of $T_{\pi(w(t))} \bbC\bbP^{N-1}$.  We can certainly combine (\ref{redSymbolPropa}) with 
Corollary \ref{prop:SameReduction} to obtain the symbol of $V (t)(\Psi_{A,w})$.  However,
in general this symbol lives in a different space than the symbol of $\Psi_{A,w}$.  
It is true that, since the entire construction of reduction is covariant with respect
to the $\text{U}(N)$ action which is transitive on the projective space, 
for a given $t$ we can apply an element of $\text{U}(N)$ and
rotate $w(t)$ back to the initial $w$.   However this element is not unique.

For this reason, we will examine the special case when
\begin{equation}\label{wolog}
\varpi = \pi(w)\quad  \text{is a critical point of }\ h, \text{and}\ h(\varpi)= 0.
\end{equation}
As we have seen in Lemma \ref{lemma:dynamicalImplications}, these assumptions in particular imply that 
$w$ is a critical point of $H:\bbC^N\to \bbR$.

We can then state:
\begin{theorem}\label{thm:SecondPropaThm}
	Under the assumption (\ref{wolog}), for each $t\in\bbR$ the symbol of $V (t)(\Psi_{A,w})$
	is equal to $\text{Mp}(\varphi_t)(\sigma_A)$, 
	where $\sigma_A$ is the symbol of $\Psi_{A,w}$,
	$\varphi_t: T_\varpi\bbC\bbP^{N-1} \to T_\varpi\bbC\bbP^{N-1}$
	is the flow of the Hessian of $h$ at $\varpi$, and $\text{Mp}$ is the metaplectic representation
	in the Bargmann space of $T_\varpi\bbC\bbP^{N-1}$.
\end{theorem}
\begin{proof}
	We will apply Proposition \ref{prop:MpIsNat}, with $E = T_w\bbC^N$, $\calC = T_w S^{2N-1}$ and 
	$\Phi: E\to E$ equal to the differential at $w$ of the time $t$ map of the
	Hamilton flow of $H$, $\Phi = d(\phi_t)_w$.  Let us identify the various relevant subspaces
	of $E$.  One has
	\[
	\calC^\circ =\{siw\;;\; s\in\bbR\}\quad\text{and}\quad J(\calC^\circ) = 
	\{sw\;;\; s\in\bbR\},
	\]
	and $\calH$ is the horizontal subspace
$\calH = \left(\bbC w\right)^\bot$,
	where the orthogonal is with respect to the standard Hermitian form on $\bbC^N$.
	The reduction $\calC/\calC^\circ$ 
	is naturally identified with $W$ and with $T_\varpi\bbC\bbP^{N-1}$.  
	Finally, observe that $\calH^\circ = \bbC w$.

	We need to verify that the hypotheses (1) and (2) of Proposition \ref{prop:MpIsNat} are satisfied.
	This follows by Lemma \ref{lemma:dynamicalImplications}, because $\Phi$ is the time $t$ map of
	the Hamilton flow of the Hessian of $H$ at $w$.
	Therefore Proposition \ref{prop:MpIsNat} applies to the present situation, which concludes the proof in view of Theorem \ref{thm:FirstPropaThm}.
\end{proof}

\subsection{Some examples}

Let us look at some examples of propagation with $N=2$.  Let $L_j:\bbC^2\to\bbR$ be given by
\begin{align*}
L_1 &= \Re(z_1\zbar_2) = \frac{1}{2} (q_1q_2 + p_1 p_2),  \\
L_2 &=\Im(z_1\zbar_2) =  \frac{1}{2} (q_1p_2 - p_1 q_2),  \\
L_3 &= \frac{1}{2} \left(|z_1|^2- |z_2|^2\right) = \frac{1}{4} (q_1^2+p_1^2 -q_2^2 - p_2^2)
\end{align*}
where we have let $z_j = \frac{1}{\sqrt{2}}(q_j-ip_j)$.  Then
$\PB{L_1}{L_2} = L_3$ and cyclic permutations.  These functions are the components of the moment map of the
$\SU(2)$ action on $\bbC^2$ with respect to the standard Pauli matrices, and they all commute with the 
circle action.  Therefore they descend to smooth functions
\[
\ell_j: \bbC\bbP^1\to\bbR
\]
which are the components of the $\SU(2)$ Hamiltonian action on the complex projective line.  Since the $L_j$
are quadratic, they are the canonical lift of the $\ell_j$.  

Using the coordinate $\zeta = z_2/z_1$ and writing $\zeta = x+iy$, the $\ell_j$'s are defined as 
\begin{align*}
\ell_1 := \Re \frac{ \zeta}{1+|\zeta|^2}, \qquad
\ell_2 :=\Im \frac{ \zeta}{1+ |\zeta|^2},
\qquad
\ell_3 := \frac{1}{2} \frac{|\zeta|^2 -1}{|\zeta|^2+1}.
\end{align*}

Each $\ell_j$ has two critical points.  Since the $\SU(2)$ action is an isometry, the Hessian
of $\ell_j$ at any fixed point $\varpi$ generates a unitary transformation of $T_\varpi \bbC\bbP^1$, which is simply
a rotation.  Under the quantum propagation of $\hat{\ell}_j$ a squeezed state at $\varpi$ simply rotates, and 
its symbol does as well.

More interesting is the action of e.g.
\begin{equation}\label{}
	h = a^2\ell_1^2 - b^2\ell_2^2, \quad a, b\geq 0 .
\end{equation}
The point $\varpi = \pi(1,0)$ is a critical point of $h$, and $h(\varpi)=0$. To apply Theorem
\ref{thm:SecondPropaThm} we need to identify the Hessian of $h$ at $\varpi$.

Using the approximation $\frac{1}{1+|\zeta|^2}\sim 1 -|\zeta|^2$, one readily checks that the Taylor expansion
of $h$ at the origin begins with
\[
h (\zeta) \sim \frac{(a^2+b^2)}{4} (\zeta^2 + \overline{\zeta}^2) + \frac{a^2-b^2}{2}\zeta\overline{\zeta}.
\]
Let us now choose $a=b=1/\sqrt{2}$ so that $h\sim \frac 14 (\zeta^2 + \overline{\zeta}^2)$.  
If $\mathfrak{z}$ is a complex coordinate on $T_\varpi \bbC \bbP^1$, the symbol 
$\sigma(\mathfrak{z},t) = f(\mathfrak{z},t)e^{-|\mathfrak{z}|^2/2}$ of a propagated squeezed
state centered at the origin solves the Schr\"odinger equation
\begin{equation}\label{schrodinger}
i \frac{\partial f(\mathfrak{z},t)}{\partial t} = 
\frac{1}{4}\left(\mathfrak{z} ^2 + \frac{d^2}{d\mathfrak{z}^2}\right) \, f(\mathfrak{z},t).
\end{equation}
We choose the time-evolved ansatz to be $f(\mathfrak{z},t)= \nu(t) \, e^{\mu(t) \mathfrak{z}^2/2}$.
We can now apply Theorem \ref{thm:EuclideanProp} (which now gives an exact solution) with $R = 1/4$ and $S=0$,
and conclude that $\nu$ and $\mu$ satisfy
\[
\dot{\mu} = \frac{1}{2i}(1+\mu^2)\quad\text{and}\quad \dot{\nu} = -\frac{i}{4}\mu\,\nu.
\]
Let us impose the initial conditions $\mu(0)=0$ and $\nu(0) = 1/(\pi\sqrt{2})$, which 
correspond to the symbol of the standard SU$(2)$ coherent state at the origin.
We find that the solutions to these ODEs are 
\begin{align*}
\mu(t) = -i\tanh(t/2), \qquad \qquad 
\nu(t) = \frac{1}{\pi \sqrt{2}} \frac{1}{\sqrt{\cosh(t/2)}},
\end{align*}
and therefore
\begin{equation*}
\sigma(\mathfrak{z},t) = \frac{1}{\pi \sqrt{2}} \frac{1}{\sqrt{\cosh(t/2)}} e^{-i \tanh(t/2)\mathfrak{z}^2/2} \, e^{-|\mathfrak{z}|^2/2}.
\end{equation*}
Making reference to the 
standard squeezed states \eqref{ketMuNormalized}, we can conclude that in this case
\begin{equation}\label{example-propagation}
e^{-ikt\hat{h}}\ket{o,0} = \nu(t)\ket{o,\mu(t)} \left(1 + O(1/\sqrt{k})\right)
\end{equation}
where the functions $\nu(t)$ and $\mu(t)$ and the Hamiltonian $\hat{h}$ are as above.

\medskip
Figure \ref{fig:propagationcoeffs} compares numerically
the left-hand-side and right-hand-side of \eqref{example-propagation} for $k=30$ and $t=2$. 
In order to compute the left-hand-side of \eqref{example-propagation}, we have written the quantum Hamiltonian $\hat{h}$ as
\begin{equation*}
	\hat{h} = a^2 \hat{L}_1 -b^2 \hat{L}_2
\end{equation*} 
where, as matrices in the basis of \eqref{ketn}, $\hat{L}_1$ and $\hat{L}_2$ are given by
\begin{align*}
	\hat{L}_1 \ket{n} &= \frac{1}{2k}\left[\sqrt{n(k-n+1)} \: \ket{n-1} + \sqrt{(k-n)(n+1)} \: \ket{n+1}\right] \\
	\hat{L}_2 \ket{n} &= \frac{i}{2k}\left[\sqrt{n(k-n+1)} \: \ket{n-1} - \sqrt{(k-n)(n+1)} \: \ket{n+1}\right] 
\end{align*}
for $n=0, \dots, k$. Notice that these matrices only have nonzero entries along the sub-diagonal and the super-diagonal. These matrices can be found using the operators in Lemma 3.2 and Lemma 3.4 in \cite{BGPU}.

\begin{figure}[h!]
	\centering
	\includegraphics[width=0.6\linewidth]{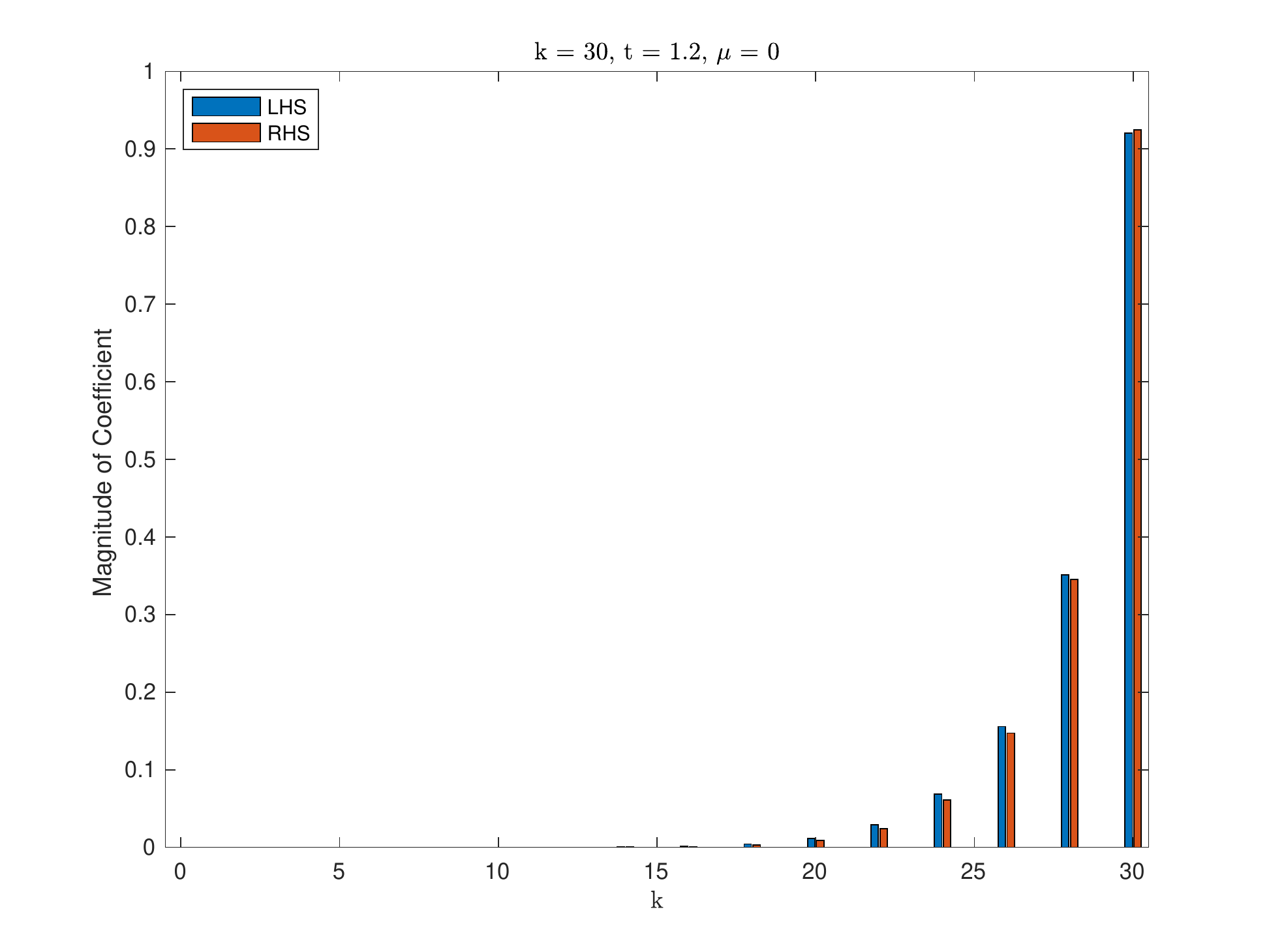}
	\caption{Plot of the magnitudes of the components of the normalized vectors on 
		both sides of \eqref{example-propagation} for $k=30$ and $t=1.2$. The difference in the $\ell^2-$norm is $|\text{LHS} - \text{RHS}| \approx 1.47 \times 10^{-2}$.}
	\label{fig:propagationcoeffs}
\end{figure}

\section{Final comments}

Since not every K\"ahler manifold is the reduction of a Euclidean space,
one can wonder how to construct squeezed coherent states, in general.

Let $\calL\to X$ be a holomorphic line bundle quantizing a  K\"ahler manifold $X$.  
The Bergman projector is the orthogonal projection
\[
\Pi_k: L^2(X,\calL^{\otimes k}) \to \calBk_X.
\]
Given $\varpi\in X$, it is easy to construct sequences of smooth sections of $\calL^k$
concentrating at $\varpi$, for example, a Gaussian in adapted coordinates times an adapted
section, in the sense of \S 3.3.  One can then apply $\Pi_k$ term-by-term to that sequence.
The resulting sequences of holomorphic sections (as well as the original sequence of smooth sections
and the Bergman kernel itself)
are special kinds of {\em isotropic functions} in the sense of \cite{GUW}. 
The symbol calculus follows from the general theory in {\em op.cit.}.  The states that
we have studied here could also have been constructed this way.

\appendix
\section{Propagation of coherent states in Bargmann space}

Here we sketch a derivation of a theorem on the propagation of Gaussian coherent states in Bargmann space.
We follow the approach of \cite{CR}, Chapter 4.

\subsection{Translations}
Let $a = (a_1,\ldots, a_N)$ and $a^* = (a_1^*,\ldots, a_N^*)$ be the (vectors of the) creation
and annihilation operators.  In Bargmann space, these are
\[
a_j = \h \frac{\partial\ }{\partial z_j}\quad \text{and}\quad a_j^* = \text{multiplication by } z_j.
\]
It is clear that $[a_j, a^*_k] = \delta_{jk}\h\, I.$
The position and momentum operators are
\begin{align*}
\widehat{Q} := \frac{1}{\sqrt{2}} (a^* +a) ,  \qquad
\widehat{P} := \frac{i}{\sqrt{2}} (a^* -a). 
\end{align*}	
Then the quantum translation by $w$ (or Weyl operator) 
\begin{equation}\label{}
\widehat{T}_w = \exp\left(i\h^{-1}\left[p\cdot\widehat{Q}- q\cdot\widehat{P}\right]\right),
\end{equation}
where $w = \frac{1}{\sqrt{2}}(q-ip)$ and $\h = 1/k$, is
$\widehat{T}_w = e^{\hinv\left(\wbar\cdot a^* - w\cdot a\right)}$,
which can be seen to be equal to
\begin{equation}\label{qTransl}
\widehat{T}_w = e^{-|w|^2/2\h} e^{\hinv \wbar\cdot a^*}e^{-\hinv w\cdot a}.
\end{equation}
This is equivalent to $\widehat{T}_w(f)(z) = e^{-|w|^2/2\h} e^{z\wbar/\h} f(z-w)$, an expression we have used before.

\medskip	
Let $t\mapsto w(t)$ be any smooth curve.  Below it will be necessary to have a formula for
$\frac{d\ }{dt}\widehat{T}_{w(t)}$. 
\begin{lemma}
	\begin{equation}\label{timeDerTransl}
	\frac{d\ }{dt}\widehat{T}_{w(t)}= \hinv
	\widehat{T}_{w(t)}\left[\frac 12 \left(w\cdot\dot{\wbar}- \dot{w}\cdot\wbar\right)
	+ \dot{\wbar}\cdot a^*- \dot{w}\cdot a\right].
	\end{equation}
\end{lemma}
\begin{proof}
	We will use (\ref{qTransl}).  By the product rule, we get the sum
	of three terms, one for each factor. The derivative of the middle factor is
	\[
	\frac{d\ }{dt} e^{\hinv \wbar\cdot a^*} = \hinv e^{\hinv \wbar\cdot a^*} \dot{\wbar}\cdot a^*.
	\]
	We want to commute $\dot{\wbar}\cdot a^*$ with the third factor.
	One can show that
	\begin{equation}\label{claimComm}
	\left[\dot{\wbar}\cdot a^*, e^{-\hinv w\cdot a} \right] =( w\cdot\dot{\wbar}) e^{-\hinv w\cdot a}.
	\end{equation}
%
	Collecting terms we get that the left-hand side of (\ref{timeDerTransl}) is
	\[
	\hinv\widehat{T}_w
	\left[
	-\frac 12 \left(\dot{w}\cdot\wbar + w\cdot\dot{\wbar} \right) + w\cdot\dot{\wbar} + \dot{\wbar}\cdot a^* - \dot{w}\cdot a
	\right].
	\]	
\end{proof}

We will also need:
\begin{lemma} The translation operator acts on the annihilation and creation operators in the following manner:
	\begin{align*}
	\widehat{T}_w \, a \, \widehat{T}_w^{-1} &=  a -\wbar I \\
	\widehat{T}_w \, a^* \, \widehat{T}_w^{-1} &=  a^* -w I.
	\end{align*}
\end{lemma}
\noindent
	The proof follows directly by calculating $\left(\widehat{T}_w \, a \, \widehat{T}_w^{-1} \right)(f)(z) $ using \eqref{qTransl}.
	The formula for the creation operator is found by taking conjugates. 
	
	\subsection{Quadratic Hamiltonians and Mp representation}
	
	The most general quadratic quantum Hamiltonian in $\bbC^N$ obtained by Weyl quantization is given by
	\begin{equation}
	\widehat{\calQ} =  a^* R (a^*)^T +  a^* S a^T + \h\frac{\tr(S)}{2} +  a \bar{R} a^T,
	\label{eq:quantum_hamiltonian}
	\end{equation}
	where 
	 $\h = 1/k$ 
	and $R$ and $S$ are $N \times N$ complex matrices with $R^T = R$ and $\bar{S}^T=S$. This operator acts on 
	$\psi(z) =f(z)e^{-k|z|^2/2}$ by acting on $f$.
	The corresponding classical Hamiltonian (the principal symbol of $\widehat{\calQ}$) is 
	the real quadratic form
	\begin{equation}\label{}
	\calQ (z) =2\Re( zRz^T) + \zbar S z^T.
	\end{equation}
	
	Let $A\in\calD_N$.  We will take $R$ and $S$ to be time-dependent (this is needed below).  
	We are interested in solving the initial value problem 
	\begin{equation}\label{quadIVP}
	i\h \frac{\partial \psi}{\partial t} =  \widehat{\calQ}(t)\psi,\qquad \psi|_{t=0} = \psi_{A,0}.
	\end{equation}
	Note that the origin is a fixed point of the Hamilton field of $\calQ$.
	\begin{proposition}\label{prop:QuadPropa}
		The solution of (\ref{quadIVP}) is
		\begin{equation}\label{quadIVPSol}
		\psi = \nu(t)\psi_{A(t), 0},
		\end{equation}
		where $A(t)$ and $\nu(t)$ solve (\ref{Adot}) and (\ref{cdot})
%
		with $A(0)=A$ and $\nu(0)=1$.
	\end{proposition}
	\begin{proof}
		We make the ansatz that $\psi$ is of the form (\ref{quadIVPSol}) and substitute into the equation.
		After some calculations we obtain the desired equations for $A(t)$ and $\nu(t)$.
	\end{proof}

\subsection{Hamiltonians of degree at most two}

Let us now consider an arbitrary Hamiltonian $H:\bbR^{2N}\to\bbR$, $t\mapsto w(t)$ a trajectory of $H$.  
For each $t$, let us write the Taylor approximation of degree at most two centered at $w(t)$, in complex coordinates:
\begin{equation}\label{}
H(z)=
H(w(0)) + (z-w(t))\frac{\partial H}{\partial z}(w(t)) +(\zbar-\wbar(t))\frac{\partial H}{\partial \zbar}(\wbar(t))
+ \calQ(t)(z-w(t), \zbar-\wbar(t))
\end{equation}
where $\calQ$ is the time-dependent Hamiltonian associated to half the Hessian of $H$ at $w(t)$,
\begin{equation}\label{}
\calQ(t)(\zeta,\zetabar) =\frac 12 \left(
\zeta H_{zz}\zetabar^T + \zeta H_{\zbar\,\zbar}\zetabar^T
+ 2\zeta H_{z\zbar} \zetabar^T
\right)
\end{equation}
where the partial derivatives are evaluated at $w(t)$.


Now let $\widehat{H}_2$ denote the Weyl quantization of $H_2$, and let $U_2(t)$ denote its propagator
with $U_2(0)= I$.
We can express $\wh{H_2}$ in terms of annihilation and creation operators as:

	\begin{equation}
	\widehat{H}_2(t) = H(w(t))  + \left(a^* -w(t) I\right) \cdot \frac{\partial H}{\partial w}(w(t)) + \left(a-\wbar_t I\right) \cdot \frac{\partial H}{\partial \wbar}(w(t)) + \widehat{\calQ}(a^* -w(t), a-\wbar_t I).
	\end{equation}
It turns out one can compute $U_2(t)$, in the following sense:
\begin{proposition}\label{prop:39} (Proposition 39 in \cite{CR})  Let $U_\calQ(t)$ be the propagator 
	of $\widehat{\calQ}$ (a metaplectic operator) satisfying $U_\calQ(0)= I$.  
	Then
	\begin{equation}\label{uTwo}
	U_2(t) = e^{i\h^{-1}\delta_t}\widehat{T}_{w(t)}\circ U_\calQ(t)\circ \widehat{T}_{w(0)}^{-1}
	\end{equation} 
	where
	\begin{equation}\label{eq:delta_t}
	\delta_t = - t H(w(0)) + \frac{i}{2} \int_{0}^{t} \left(w(s) \dot{\bar{w}}(s) - \dot{w}(s) \bar{w}(s)\right).
	\end{equation}
\end{proposition}
\begin{proof}
	Denote for now the right-hand side of (\ref{uTwo}) by $U_2$.  The proof is to show that 
	\begin{equation}\label{}
	i\h \dot{U}_2 = \widehat{H}_2 U_2\quad\text{and}\quad U_2(0) = I.
	\end{equation}
	The second condition is clearly satisfied, so let's differentiate the right-hand side of (\ref{uTwo}). 
	We get:
	\[
	\dot{U}_2 = -i\hinv \dot{\delta_t}\, U_2+  (\rm{II}) + \rm(III),
	\]
	where (using (\ref{timeDerTransl}))
	\[
	\text{(II)}=  \hinv e^{i\h^{-1}\delta_t}
	\widehat{T}_{w(t)}\left[\frac 12 \left(w\cdot\dot{\wbar}- \dot{w}\cdot\wbar\right)
	+ \dot{\wbar}\cdot a^*- \dot{w}\cdot a\right] U_\calQ(t)\widehat{T}_{w(0)}^{-1}
	\]
	and
	\[
	\text{(III)} = -i\hinv e^{i\h^{-1}\delta_t}\widehat{T}_{w(t)}\widehat{\calQ}(t) U_\calQ(t) \widehat{T}_{w(0)}^{-1}.
	\]
	Using again the definition of $U_2$ to solve for $U_\calQ\widehat{T}_{w(0)}^{-1}$, we can write
	\[
	\text{(II)} = \hinv\widehat{T}_{w(t)}\left[\frac 12 \left(w\cdot\dot{\wbar}- \dot{w}\cdot\wbar\right)
	+ \dot{\wbar}\cdot a^*- \dot{w}\cdot a\right]  \widehat{T}_{w(t)}^{-1} U_2
	\]
	and
	\[
	\text{(III)} = -i\hinv  \widehat{T}_{w(t)}\widehat{\calQ}(t)  \widehat{T}_{w(t)}^{-1} U_2.
	\]
	We analyze (II) further, the key step being
	\[
	\widehat{T}_{w(t)}\left[\dot{\wbar}\cdot a^*- \dot{w}\cdot a\right]  \widehat{T}_{w(t)}^{-1} = 
	\widehat{T}_{w(t)}(\dot{\wbar}\cdot a^*) \widehat{T}_{w(t)}^{-1}  - \widehat{T}_{w(t)}( \dot{w}\cdot a) \widehat{T}_{w(t)}^{-1}
	= \dot{\wbar} \cdot (a^* - wI) - \dot{w} \cdot (a-\wbar I) .
	\]
	After some calculations one finds that 
	\[
	i\h\dot{U}_2\,U_2^{-1} =  -\dot{\delta}_t +\frac{i}{2} \left(w\cdot\dot{\wbar}- \dot{w}\cdot\wbar\right) - H(w(t)) + \widehat{H}_2(t),
	\]
	so $\dot{\delta_t} = -H(w(0)) +  \frac{i}{2} \left(w\cdot\dot{\wbar}- \dot{w}\cdot\wbar\right)$ using $H(w(t)) = H(w(0))$. Integrating gives \eqref{eq:delta_t}.
\end{proof}

\begin{corollary}\label{cor:39}
\[
U_2(t)\left(\psi_{A,w(0)}\right) = \nu(t) e^{ i k \delta_t}\psi_{A(t),w(t)}
\]
where $\nu(t)$ and $A(t)$ satisfy (\ref{Adot}) and (\ref{cdot}).  
\end{corollary}

\subsection{Propagation}
First we need a preliminary estimate which we state without proof:
\begin{proposition}
	Let $\widehat{H}, \widehat{H}_2$ be semi-classical pseudodifferential operators acting on the Bargmann space
	of $\bbC^N$, with principal symbols $H$ and $H_2$.  Let $w\in\bbC^N$ and assume that $H-H_2$ vanishes at $w$,
	together with its first and second derivatives.  Then, for any $A\in\calD_N$
	\begin{equation}\label{prelimEst}
	\norm{(\widehat{H}-\widehat{H}_2)\psi_{A,w}} = \norm{\psi_{A,w}}\cdot O(\h^{3/2}).
	\end{equation}
\end{proposition}
To finish the proof of (\ref{propaUpstairs}) we follow the argument of Chapter 4 in \cite{CR}.
By Duhamel's principle
\begin{equation}\label{Duhamel}
U(t) - U_2(t) = \frac{1}{i\h}\int_0^t U(t,s)\left(\widehat{H}-\widehat{H}_2(s) \right)U_2(t,s)\, ds
\end{equation}
where $U(t,s)$ is the propagator for $\widehat{H}$ such that $U(t,t)= I$ (and similarly for $U_2(t,s)$),
it follows that
\[
\norm{U(t)(\psi_{A,w(0)}) - U_2(t)(\psi_{A,w(0)})}\leq \hinv \int_0^t 
\norm{\left(\widehat{H}-\widehat{H}_2(s) \right)\phi_{t,s}}\,ds
\]
where $\phi_{t,s} = U_2(t,s)(\psi_{A,w(0)})$.  This, combined with (\ref{prelimEst}), yields 
\begin{equation}\label{}
\norm{U(t)(\psi_{A,w(0)}) - U_2(t)(\psi_{A,w(0)})} = \norm{\psi_{A,w(0)}}\cdot O(\h^{1/2}).
\end{equation}
Using Corollary \ref{cor:39} we obtain Theorem \ref{thm:EuclideanProp}.

\nocite{*} 
\bibliographystyle{plain} 
\bibliography{coherent-states-ref.bib}


%
%
%
%
%

\end{document}